\documentclass[journal]{IEEEtran}

\usepackage{graphicx}
\usepackage{color}
\usepackage{rotating}
\usepackage{tabularx}
\usepackage{pdflscape}
\usepackage{amsmath,amsthm,amssymb}
\usepackage{blkarray}
\usepackage{algorithm,algorithmic}
\usepackage{bbm,dsfont}
\usepackage{caption,subcaption}
\usepackage{wrapfig}
\usepackage{cancel}
\usepackage{enumerate, cases}
\usepackage{thmtools,thm-restate}
\usepackage{mathtools}
\usepackage{cite}

\usepackage[none]{hyphenat}
\usepackage{multirow}
\usepackage{makecell}
\usepackage{setspace}

\usepackage{authblk}
\usepackage{dblfloatfix}
\usepackage{array}
\usepackage{enumitem}
\usepackage{tikz}
\usetikzlibrary{decorations.pathreplacing}
\usepackage{pgfkeys}
\usetikzlibrary{intersections}

\DeclareMathOperator*{\argmax}{arg\,max}

\newcolumntype{C}[1]{>{\centering}m{#1}}
\newcommand{\EE}[1]{\mathbb{E}\left[#1\right]}

\theoremstyle{plain}
\newtheorem{thm}{Theorem}

\newtheorem{cor}{Corollary}

\newtheorem{rem}{Remark}
\newtheorem{defi}{Definition}
\newtheorem{assu}{Assumption}

\begin{document}

\title{UB3: Best Beam Identification in Millimeter Wave Systems via Pure Exploration Unimodal Bandits}


\author[1]{Debamita Ghosh}
\author[2]{Haseen Rahman}
\author[3]{Manjesh K. Hanawal}
\author[4]{Nikola Zlatanov}
\affil[1]{IITB-Monash Research Academy, IIT Bombay, India, email: debamita.ghosh@iitb.ac.in}
\affil[2]{MLiONS Lab, IEOR, IIT Bombay, India, email: haseenrahman@gmail.com}
\affil[3]{MLiONS Lab, IEOR IIT Bombay, India, email:mhanawal@iitb.ac.in}
\affil[4]{Innopolis University, Russia, email: n.zlatanov@innopolis.ru}


\maketitle

\begin{abstract}
Millimeter wave (mmWave) communications have a broad spectrum and can support data rates in the order of gigabits per second, as envisioned in 5G systems. However, they cannot be used for long distances due to their sensitivity to attenuation loss. To enable their use in the 5G network, it requires that the transmission energy be focused in sharp pencil beams. As any misalignment between the transmitter and receiver beam pair can reduce the data rate significantly, it is important that they are aligned as much as possible. To find the best transmit-receive beam pair, recent beam alignment (BA) techniques examine the entire beam space, which might result in a large amount of BA latency. Recent works propose to adaptively select the beams such that the cumulative reward measured in terms of received signal strength or throughput is maximized. In this paper, we develop an algorithm that exploits the unimodal structure of the received signal strengths of the beams to identify the best beam in a finite time using pure exploration strategies. Strategies that identify the best beam in a fixed time slot are more suitable for wireless network protocol design than cumulative reward maximization strategies that continuously perform exploration and exploitation. Our algorithm is named Unimodal Bandit for Best Beam (UB3) and identifies the best beam with a high probability in a few rounds. We prove that the error exponent in the probability does not depend on the number of beams and show that this is indeed the case by establishing a lower bound for the unimodal bandits. We demonstrate that UB3 outperforms the state-of-the-art algorithms through extensive simulations. Moreover, our algorithm is simple to implement and has lower computational complexity.

\end{abstract}

\begin{IEEEkeywords}
mmWave, Bandit learning, pure exploration
\end{IEEEkeywords}

\section{Introduction}\label{sec:intro}
There is a growing demand for higher data rates with the advent of emerging data-intensive applications like virtual reality, mobile gaming, and HD quality video streaming. The wireless networks have improved in terms of data rates but are still constrained by the available bandwidth in the sub-6 GHz spectrum to meet the data rates required for emerging applications. The millimeter wave (mmWave) band, with a spectrum ranging from 30 GHz to 300 GHz, offers an abundant spectrum and can support data rates of gigabits per second that are envisioned in 5G networks. Significant efforts in the standardisation of mmWave systems, such as IEEE 802.11ad \cite{IEEEStd2012, TWC2017_zhou2017enhanced} and ongoing IEEE 802.11ay \cite{Magazine2017_802.11ay}, are underway, and the 5G networks with mmWave systems are on the path of commercialisation through extensive field trials. 

Though mmWave systems offer higher data rates, they come with a set of challenges—the small wavelengths of mmWaves make them suffer significant attenuation, resulting in a rapid deterioration of signal strength. Thus, unlike traditional terrestrial communication systems, mmWave communication requires highly directional communication with energy focused on narrow beams to achieve the required signal strengths at the receiver. Some challenges in using mmWave systems in mobile communications are highlighted in standards \cite{IEEEStd2012, Magzine2014_80211ad}. However, on the brighter side, small wavelengths allow transmission antennas with small form factors to be packed closely and focus signal energy in specific directions, forming sharp beams.

The other challenge in mmWave communication is that transmitter and receiver beams need to be aligned before data transfer; otherwise, any advantage of high data rates is not realised. A few degrees of misalignment in the beam directions between transmitter and receiver can reduce the data rates from gigabits per second to a few megabits per second, jeopardising the gain from the high spectrum of mmWave systems \cite{INFOCOM2015_mmwavesteering, IEEEStand2012_ieee, INFOCOM2018_EfficientBeamAllignment}. This gives rise to the problem of beam alignment (BA), where one needs to find the best transmitter and receiver beam pair that provides the best rates. The BA problem is critical to building better 5G communication networks with mmWave systems. Our goal in this paper is to learn the best beam pair in a given number of slots with a high probability.

One naive approach to performing BA is an exhaustive search of all available beams at the base station (BS) and user equipment (UE). This strategy does not scale well because it has the complexity of order $\mathcal{O}(K^2)$, where $K$ is the number of beams at BS and UE. The IEEE standard 802.11ad \cite{IEEEStd2012} decouples the BA by performing it in two stages. In the first stage, the BS uses a quasi-omnidirectional beam while the UE scans through its beams to identify the best one. In the next stage, the UE uses a quasi-omnidirectional beam while the BS scans through its beam to identify the best one. This search strategy has a complexity that is linear in $K$. However, as discussed in \cite{SIGMETRICS2015_60GHzIndoor, MobileComputing2014_Demistifying60GHz}, this method can still take time in seconds. 

The initial access (IA) phase in 5G mmWave provides a mechanism to identify beam directions between a BS and UE \cite{TWC2016_InitialAcess, TWC2020_InitialAcessBeamSweeping, INFOCOM2019_InitialAccessSmartlink, TON2018_IntialAcessFastND}. BSs periodically use the IA phase to discover new UEs and check if the best beam for already existing UEs has changed \cite{3GPP_5GStd}. During the IA, BS transmits synchronization signals that UE can measure and report back the received signal strengths (RSS). IA can be used to explore and identify the best pair, as no data is transferred in this phase. Further, BS can adapt to the non-stationary environment by periodically rerunning the IA phase, where the periodicity can depend on the mobility rate and the atmospheric conditions.

We address the BA and user tracking issues in mmWave using the fixed-budget pure exploration Multi-Armed Bandit (MAB) framework, where pure exploration is performed in the IA phase. The BA problem has already been addressed using the MAB framework, which uses cumulative regret minimization algorithms that balance exploration and exploitation to find the best beam pair \cite{INFOCOM2018_EfficientBeamAllignment, hba, Infocom2020_MAMBA}. However, the stopping time in these algorithms is random, leading to the following difficulties: 1) The learning phase may extend beyond the IA phase and reduce the number of time slots available for data transfer. 2) Due to continuous exploration, sub-optimal beams can be used for data transfer, resulting in outages. To overcome these issues, we complete the learning phase within a fixed number of time slots (budget) of the IA phase using adaptive exploration. Moreover, we exploit the structural properties of the RSS across the beams to accelerate the learning process. 

Several studies validate that the RSS of the beams in mmWave systems follows a multi-modal structure, with one peak corresponding to the line-of-sight path and others corresponding to the non-line-of-sight paths \cite{hba, INFOCOM2018_EfficientBeamAllignment}. Often, there is one dominant peak, and the multi-modal functions can be treated as unimodal. Bandits with a unimodal structure are well studied in the literature in the cumulative regret setting with optimal algorithms \cite{ICML2014_UnimodalBandits, ICML2011_UnimodalBandits, Arxiv2020_ExplorationFreeUnimodalBandits}. However, the fixed-budget pure exploration bandit with unimodal structure is not well studied, and optimal algorithms are not known. In this work, we develop a new fixed-budget pure exploration algorithm that exploits the unimodal structure. The new algorithm is named {\it{\ref{algo:UB3}}} and is based on the idea of sequential elimination of sub-optimal beams. {\it{UB3}} achieves an error probability of the order of $\mathcal{O}(\log K \exp(-T_1D^2_L))$ after $T_1$ time slots of IA phase, where $D_L$ is the minimum gap between two successive means of the arms. When no unimodal structure is assumed (unstructured), the best known achievable error probability is $\mathcal{O}(\log K\exp(-T_1/H\log K)) $ \cite{audibert2010best}, where $K$ is the number of beams and $H$ is the problem-dependent constant. Thus, by exploiting unimodal structure, we achieve a better error probability where the error exponent does not depend on the number of beams, and demonstrate that this is anticipated by establishing a lower bound for unimodal bandits. Extensive simulation on realistic wireless networks demonstrates 
{\it{UB3}} identifies the best arm with a probability more than 95\% within $100$ time slots for $16$ beams, while the other state-of-the-art algorithms need $500$ time slots. This translates into throughput gains of more than 15\% compared to other algorithms. Moreover, {\it UB3} does not require any prior knowledge of channel fluctuations. In summary, our contributions are as follows:
\begin{itemize}
    \item We set up the problem of beam alignment in mmWave systems as a fixed-budget pure exploration multi-armed bandit problem in Sec. \ref{sec:setup}. 
    \item We exploit the unimodal structure of the RSS of the beams and develop an algorithm named {\it{\ref{algo:UB3}}} to identify the best beam with a high probability in fixed time slots or within the IA phase of the BA in Sec. \ref{sec:algo}.
    \item We provide an upper bound on the error probability of {\it{UB3}} in identifying the best beam and show that the error exponent does not depend on the number of beams. We demonstrate that is anticipated by establishing a lower bound for unimodal bandits in Sec. \ref{sec:analysis}.
    \item We perform extensive simulations to validate the superior performance of {\it{UB3}} compared to other state-of-the-art algorithms like {\it HOSUB}, {\it{HBA}}, {\it{LSE}} and {\it{Sequential Halving}} in Sec. \ref{sec:experiments}. In agreement with the theoretical bounds, {\it{UB3}} is not affected by the number of beams.
\end{itemize}
\subsection{Related Work}
As there is growing interest in mmWave systems in academia and industry, various aspects of mmWave are being studied. For the recent advances in the mmWave systems, we refer to surveys \cite{Access2020_Survey,Springer2015_mmWaveSurvey,ComSurveys2018_mmWaveSurvey}.

Several approaches are proposed to solve the BA problem. The compressive sensing-based signal processing methods \cite{TSP2016_CpmpressSensingmmWave, Sigcom2018_FastmmWaveAlignment, TSP2019_FastBeamAlignmentSparseCoding} utilize the sparse characterization of mmWave to learn the best beam. They work better when accurate channel state information is available. \cite{TWC2019_ExploreEliminate} proposes an Optimized Two-Stage Search (OTSS) where a suitable candidate set of beams is identified in the first stage based on the received signal profile, and in the second stage, the best beam from the surviving set is selected with additional measurements. Codebook-based hierarchical methods are proposed in \cite{JSAC2009_CodeBasedBeamforming, TWC2016_CodeBookBeaming, INFOCOM2018_LinearBlockCodingMillimeterWave} which also require channel state information for BA. \cite{Access2019_PositionAidedmmWaveBeamTraining, INFOCOM2018_FastMachineLeaning} utilizes the location information to perform fast BA, which is feasible only when the location information of the UE is available at the BS. \cite{GlobalSIP2016_BeamTracking, ICC2016_TrackingAnlgesofDepartureandArrival} use Kalman filters to detect the angles of arrivals and departures to track UE. Recently, machine learning \cite{TCCN2020_MachineLearningAssisted}, and deep learning \cite{ICASSP2020_DeepLearningBeamAlignment, TWC2022_DeepReinforcementLearning} methods have been used for BA, which requires offline training of the models. 

Our work is closer to \cite{INFOCOM2018_EfficientBeamAllignment, Infocom2020_MAMBA, hba, ICC2021_HOSUB} which uses an online learning approach to optimize the BA problem. The authors in \cite{INFOCOM2018_EfficientBeamAllignment} develop an algorithm named {\it{Unimodal Beam Alignment (UBA)}} that exploits the unimodal structure of received power. The algorithm is built on the OSUB algorithm \cite{ICML2014_UnimodalBandits} by adding stopping criteria. The algorithm assumes that the mean powers are known, and the stopping criteria is based on these mean powers. The {\it{Hierarchical Beam Alignment (HBA)}} \cite{hba} algorithm also exploits the unimodal/multimodal structures of beam signal strengths to narrow down on the optimal beam. {\it HBA} has shown better performance for beam identification than {\it UBA}. However, for $T$ time slots, the computational complexity of {\it{HBA}} is $O(T^2)$, as in each time slot, the algorithm restarts the search, making the running time linear in each time slot. Moreover, {\it HBA} requires knowledge of channel fluctuations, which is not practical. The {\it Hierarchical Optimal Sampling of Unimodal Bandits (HOSUB)} \cite{ICC2021_HOSUB} exploits the benefits of hierarchical codebooks
and the unimodality structure of the beam signal strengths to achieve fast
beam steering. Simulations show better performance of {\it HOSUB} compared to {\it HBA}, as well as a large reduction in computational complexity. However, the authors in \cite{ICC2021_HOSUB} did not provide any theoretical guarantees on their proposed algorithm. The authors in \cite{Infocom2020_MAMBA} develop an algorithm named {\it{MAMBA}} that aims to maximize the cumulative rate obtained over a period using the Thompson sampling algorithm named {\it{Adaptive Thompson Sampling (ATS)}}. Unlike the {\it UBA} and {\it HBA}, the exploration never ends in {\it ATS} and may keep selecting the sub-optimal beams. 

Our work develops an online learning algorithm for BA using a fixed-budget pure exploration setup \cite{audibert2010best,JMLR2016_ComlexityofBestArmIdentification}. We exploit the unimodal structure of beam RSS to eliminate the sub-optimal beams and narrow the beam search space quickly. Fixed-budget pure exploration strategies are more suitable for the BA problem, as the exploration can be completed in the IA phase, and no exploration is required during the data transfer, simplifying the design of protocols. To our knowledge, this has not been studied in 5G networks with mmWave systems. 

\section{Problem Setup}\label{sec:setup}
In this section, we discuss the system and channel model used in mmWave system. We follow the setup and notation similar to \cite{hba}.

\subsection{System Model}
We consider a point-to-point mmWave wireless system between a transmitter, refer as mmWave BS, and a receiver, refer as mmWave UE, in a static environment as shown in Fig.~\ref{fig:System_model}. We focus on analog beamforming and consider one ADC with an RF chain that focuses on one direction at a time. The transmitter has $K$ phased-array antennas, where each antenna has a phase shift to form a narrow directional beam. The antennas are evenly spaced by a distance $D\approx \lambda/2$ forming a uniform linear array, where $\lambda$ is the carrier wavelength. As the IEEE 802.11ad can decouple the BA, we consider that the receiver keeps a quasi-omnidirectional beam and the transmitter scans over the beam space to identify the best beam. This is a reasonable assumption due to the small form factor and fewer antennas on UEs. In the following, we focus on beam alignment (BA) at the BS side, and the extension that involves both BS and UE is straightforward. 
\begin{figure}
	\centering
\includegraphics[scale=.42]{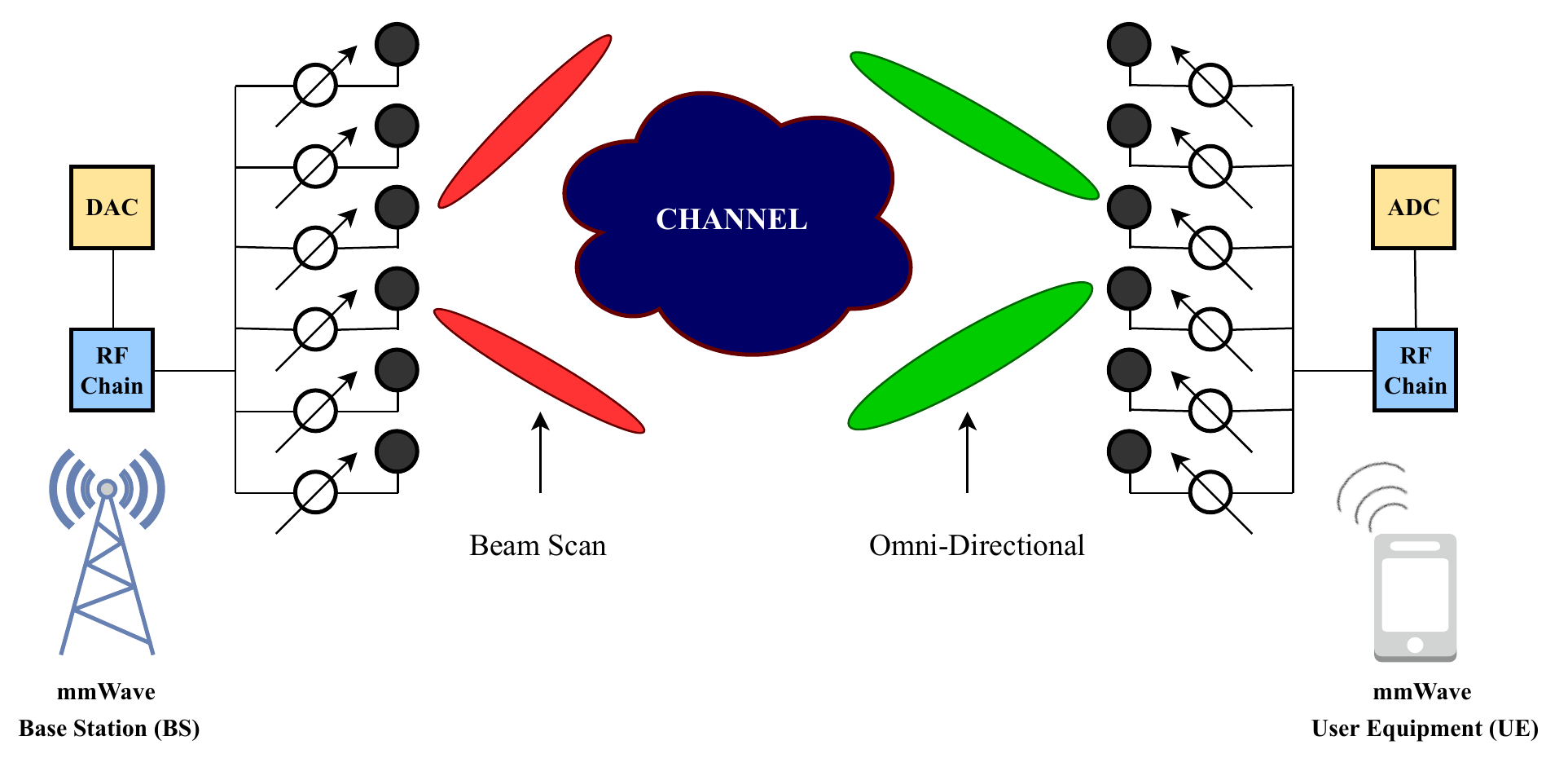}
	\caption{\small{A point-to-point mmWave Communication System.}}
	\label{fig:System_model}
\end{figure}
\begin{figure*}[t]
	\centering	\includegraphics[scale=1.5]{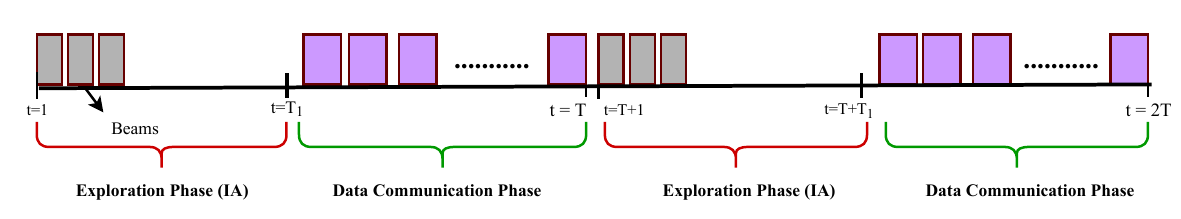}
    	\caption{\small{Beam exploration (IA) phase followed by data transfer phase.}}
	\label{fig:BA}
\end{figure*}

\subsection{Channel Model}
Due to the sparse characteristics of mmWave channel, we consider Saleh-Valenzuela channel model \cite{JSAC2014_MillimeterWaveChannelModel}. Suppose there are $L$ paths where one is the dominant line-of-sight (LOS) path and $L-1$ non-line-of-sight (NLOS) paths. Let $\mathbb{C}$ denote the set of complex numbers. The channel vector between the transmitter and the receiver is given by
\begin{align}
    {{\bf h}} = g_0{{\bf a}(v_0)} + \sum\limits_{l=1}^{L-1}g_l{{\bf a}(v_l)} \in \mathbb{C}^{K\times 1}
\end{align}
where ${\bf{a}(v)}=\left\{e^{j\frac{2\pi D}{\lambda}k v}: 0\leq k \leq K-1 \right\} \in \mathbb{C}^{K\times 1}$ denote the vector of sinusoids at spatial angle $v$, $g_0$ is the channel gain of LOS path, and $g_l, 1\leq l\leq L-1$ is the channel gains of $l^{th}$ NLOS path. $v=\cos \theta$ denotes the spatial angle of the channel associated with physical angle $\theta$. As assume that the channel remains static for a duration $T$ time slots and the channel vector keeps invariant in the BA process. 

Let $\mathbf{B} \in \mathcal{C}^{K\times K}$ denote the unitary discrete Fourier transform (DFT)  of the transmit beam space, where the $k^{th}$ column corresponds to the $k^{th}$ beam, i.e.,  
\begin{align}
{\mathbf{B}} &= [\bf{b_1},\bf{b_2},\dots,\bf{b_K}] =\frac{1}{\sqrt{K}}[{\bf a}(w_1),{\bf a}(w_2),\dots,{\bf a}(w_K)]\label{eqn:BFM}
\end{align}
where $w_k = \frac{2k - K}{K}$ denotes the spatial angle of the $k^{th}$ beam. Then, the received signal of $k^{th}$ beam (with receiver omnidirectional) is given by
\begin{align}
    y_k = \sqrt{P}{\bf{h}}^H{\bf{b_k}} + n,\label{eqn:received_signal}
\end{align}
where $n$ denotes the additive white Gaussian noise with mean noise power $N_0W$, with  noise power density $N_0$ and channel bandwidth $W$. We denote the received signal strength (RSS) of the $k^{th}$ beam as $r_k=|y_k|^2$ and denotes its mean value as $\mu_k = \EE{r_k}$.  Let $k^*=\argmax_k \mu_k$. Then the beam $b_{k^*}$ denotes the optimal beam with the best mean RSS. 

\begin{defi}{(Unimodality):}
The unimodality structure
indicates that, $\forall b_k \in  {\bf{B}},$ there exist a path such that $\mu_1 < \mu_2 <\dots< \mu_{k^*}$ and $\mu_{k^*} > \mu_{K^*+1} >\dots> \mu_K.$
\end{defi}
For the case when only LOS path is present with channel gain $g$ and spatial angle $v$, mean RSS is given as $\mu_{k} = \frac{Pg^2}{K}\delta(w_k - v) + N_o W$ for all $b_k \in {\bf{B}}$, where $\delta(x) =\frac{\sin^2(K\pi Dx/\lambda)}{\sin^2(\pi Dx/\lambda}$ denotes the antenna directivity function for angular misalignment $x$. For $b_{k}\in \mathbf{B}$, $\mu_{k}$ is a function of angular misalignment $w_k-v$ and has unimodal property \cite[Thm. 1]{hba}. In the following, we only consider the beams with the unimodal property as in \cite{INFOCOM2018_EfficientBeamAllignment} and later discuss how to extend our method to the multimodal case involving multiple NLOS of paths.

\begin{rem}
We assumed that the Modulation and Coding Scheme (MCS) on each beam is fixed. However, we can easily accommodate different MCS on each beam by treating the RSS as a vector corresponding to MCS and considering the mean rate as done in \cite{Infocom2020_MAMBA}.
\end{rem}

\subsection{Problem Formulation}
We assume a slotted system where the length of the IA phase is $T_1$ time slot. As specified in the 5G standards, we let the BS rerun IA phase periodically after every $T$ time slot where $T>T_1$. We assume that during the period $T$, the environment is stationary such that the best beam remains the same. We focus on one period between two successive reruns of the IA phase and index the time slots in that period as $t=1$ to $t=T$, refer to Fig. \ref{fig:BA}. In each time $t$, the BS can select one of the beams from set $\mathbf{B}$ and obtain as feedback the RSS at the receiver. The feedback is obtained through ACK/NACK sent back by the receiver -- BS measures the signal strength of the received ACK/NACK, which gives the RSS at the receiver \cite{Infocom2020_MAMBA}. During the IA duration of $T_1$ time slots, no information signals are transmitted and throughput or error probability is not a concern. However, during the period of $T-T_1$ data is transmitted and it is desirable to obtain high throughput in this period using the best possible beam. We are thus interested in algorithms that output an optimal beam at the end of $T_1$ time slots (IA phase).

We model the problem as a fixed-budget pure exploration multi-armed bandit \cite{audibert2010best, ICML2012_PACSubsetSelection} where the goal is to identify the optimal arm within a fixed budget with high confidence. Following the terminology of multi-armed bandits, we refer to beams as arms. A policy is any strategy that selects an arm in each time slot given the past observations. 
Let $k_t \in \mathbf{B}$ denotes the arm selected by a policy at time $t$. By playing an arm $k_t$, the policy observes the feedback $r_{k_t}$ which is a noisy RSS. The choice of $k_t$ can depend on the beams selected in the past and their associated RSS values. We assume that RSS values observed in each time slot are independently and identically distributed across the arms and time. The distribution of RSS is governed by channel fluctuations, such as shadow fading and the disturbance effect, and follows unknown fixed distributions within a time period $T$.  However, without loss of generality, we assume that values of RSS are bounded in some interval.

For a given policy $\pi$, let $\hat{k}_{T_1}^\pi$ denote the index of the arm output by $\pi$ at the end of $T_1$ time slots. Let $\Pi$ denote the set of all policies of pure-exploration that output algorithms within a fixed budget of $T_1$. Then our goal is to find a policy in $\Pi$ that minimizes the probability that the arm output at the end of $T_1$ is not an optimal arm, i.e., 

$$\min_{\pi \in \Pi}\Pr \left (b_{\hat{k}_{T_1}}^\pi \neq b_{k^*}\right ),$$
where for each policy, $\Pr(\cdot)$ is calculated with respect to the samples induced by the policy.  We note that our criteria is different from those set in {\it{UBA}} \cite{INFOCOM2018_EfficientBeamAllignment} and {\it{HBA}} \cite{hba} which aim to minimize cumulative regret. Though both {\it UBA} and {\it HBA} have stopping criteria beyond which they play a fixed arm, the stopping time can be random, making their practical implementation challenging. Whereas, the policy considered by us completes the exploration phase deterministically after time $T_1$  which makes their implementation easier in a wireless setup. Fig.~\ref{fig:BA} depicts the structure of our policy.

\section{Algorithms} 
\label{sec:algo}
 In this section, we propose an algorithm named {\it UB3-Unimodal Bandit for Best Beam} that finds the optimal beam after exploiting the unimodal structure of mean RSS within $T_1$ time slots. The algorithm is based on the {\it{Line Search Elimination (LSE)}} algorithm developed in \cite{ICML2011_UnimodalBandits}, where the algorithm samples and eliminate arms in multiple phases till one arm survives after $L+1$ phases. The pseudo-code of {\it UB3} is given in \ref{algo:UB3}. It is parameter free and only takes $K$ and $T_1$ as inputs. 
 
 {\it UB3} runs in $L+1$ phases. Arms are sampled and eliminated in each phase such that only one arm survives after the $L+1$ phase. We first explain the number of rounds in each phase. Let $N_l$, for $l=1,2,\ldots, L+1$, denotes the number of samples in phase $l$. Then,
\begin{equation}
\label{eqn:Nl}
    N_l=\begin{cases}
        \frac{2^{L-2}}{3^{L-1}}T_1 & \mbox{ for  } l=1,2\\
        \frac{2^{L-(l-1)}}{3^{L-(l-2)}}T_1 & \mbox{ for  }  l=3,4,\ldots, L+1
    \end{cases}
\end{equation}
which satisfies that
\begin{align}
\label{eqn:L_T}
    2\times \frac{2^{L-2}T_1}{3^{L-1}} + \sum_{l=3}^{L+1}\frac{2^{L-(l-1)}T_1}{3^{L-(l-2)}} = T_1.
\end{align}
After the first two phases, the number of samples increases by a factor of $3/2$ in each subsequent phase. This increase in the number of samples helps to distinguish between the empirical means of the remaining arms, which are likely to be closer.

{\it UB3} works as follows. Let ${\bf{{B}_l}}=\{b_1,b_2,\dots,b_l\}$ denote the set of arms available in phase $l$ and $j_l := |{\bf{{B}_l}}|$ is the number of arms in the set ${\bf{{B}_l}}$. In phase $l=1,2,\ldots L$, the algorithm selects four arms $\{k^M, k^A, k^B, k^N\} \in {\bf B_l}$, which include the two extremes and two middle arms uniformly spaced from them (lines 4-7). Each of the arms is sampled for $\frac{N_l}{4}$ number of times (line 8). At the end of the phase, their empirical means denoted $\hat{\mu}_i$ (line 9) are obtained as follows:
\begin{equation}
\label{eqn:empmeans}
    \hat{\mu}^l_i = \frac{1}{N_l/4}\sum\limits_{s=1}^{N_l/4}r^l_{i_s},\quad  \forall i \in \{k^M, k^A, k^B, k^N\}
\end{equation}
where $r^l_{i_s}$ denotes  the $s^{th}$ noisy RSS sample from $i^{th}$ arm in phase $l$. Based on these empirical means, we eliminate at most $1/3^{rd}$ of the number of arms from the remaining set\footnote{If the number of arms in a phase is not a multiple of $4$, then less than $1/3^{rd}$ will be eliminated in that phase.}. More specifically, if the arms $k^M$ or $k^A$ has the highest empirical means, then we eliminate all the arms succeeding $k^B$ in the set ${\bf B_l}$ (line 11 and 12).  Similarly, if the arms $k^B$ or $k^N$ has the highest empirical means, then we eliminate all the arms preceding $k^A$ in the set ${\bf B_l}$ (line 13 and 14). Fig.~\ref{fig:AlgorithmPic} gives a pictorial representation of the elimination of arms in two possible cases. The remaining set of arms are then transferred to the next phase. In phase $L+1$, we are left with three arms. Each one of them is sampled $\frac{N_{L+1}}{3}$ number of times and the one with the highest empirical mean is the output of the algorithm as the optimal arm (lines 18-23).
\begin{figure}
\vspace{-3mm}
    \centering
    \includegraphics[scale=0.55]{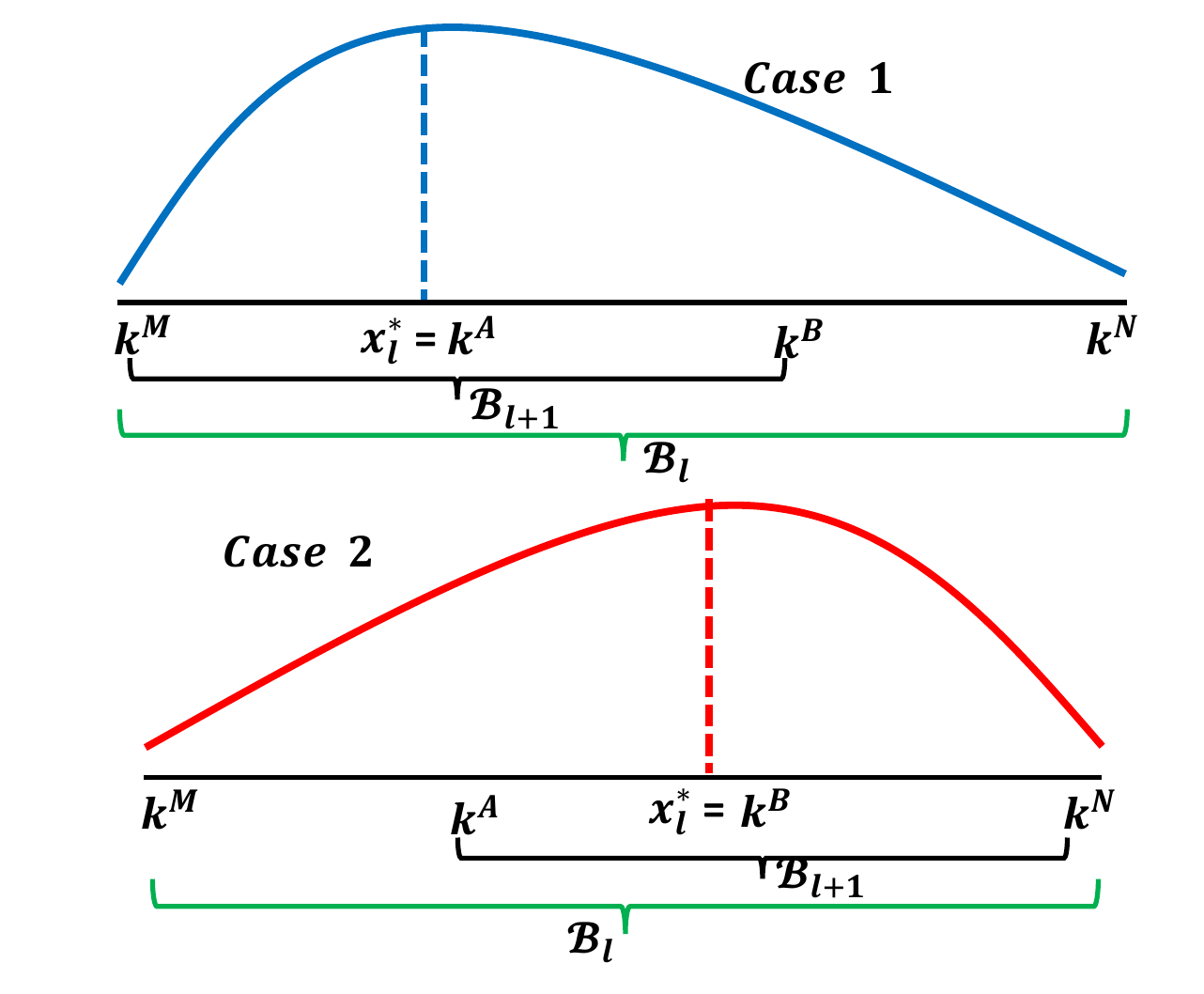}
    \caption{Different cases of elimination in phase $l$.}
    \label{fig:AlgorithmPic}
\end{figure}
\begin{algorithm}
	\renewcommand{\thealgorithm}{Unimodal Bandit for Best Beam (UB3)}
	\floatname{algorithm}{}
	\caption{\bf }
	\label{algo:UB3}
    \begin{algorithmic}[1]
        \STATE \textbf{Input:} $T_1$ and $K.$
        \STATE \textbf{Initialise:} ${\bf{{B}_1}} = {\bf B},$ $j_1\leftarrow |{\bf{{B}_1}}|$. Calculate L from \eqref{eq:alg21}.
         \FOR{$l= 1$ to $L$}
            \STATE $k^M\leftarrow$ First arm of ${\bf B_l}$;
            \STATE $k^N\leftarrow$ Last arm of ${\bf B_l}$;
              \STATE $k^A \leftarrow \lceil j_l/3\rceil^{th}$ arm of ${\bf B_l}$;
              \STATE $k^B \leftarrow \lfloor 2j_l/3\rfloor^{th}$ arm of ${\bf B_l}$;
                \STATE Sample each arm in $\{k^M, k^A, k^B, k^N\}$ for $\frac{N_l}{4}$ number of times from \eqref{eqn:Nl}
                \STATE Obtain $\hat{\mu}^l_{k^M}$, $\hat{\mu}^l_{k^A}$, $\hat{\mu}^l_{k^B}$, $\hat{\mu}^l_{k^N}$ by \eqref{eqn:empmeans}.
                \STATE $x^*_l = \argmax\limits_{i \in \{k^M, k^A, k^B, k^N\}} \hat{\mu}_i.$
                \IF{$x^*_l == \{k^M, k^A\}$}
                    \STATE ${\bf B_{l+1}} \leftarrow\{k \in {\bf B_l}: k^M \leq k \leq k^B\}$ Shrink to left
                \ELSIF{$x^*_l == \{k^B, k^N\}$}
                     \STATE ${\bf B_{l+1}}$ $\leftarrow\{k \in {\bf B_l}: k^A \leq k \leq k^N\}$ Shrink to right
                \ENDIF
                \STATE $j_{l+1} \leftarrow |{\bf B_{l+1}}|;$
               \ENDFOR
        \FOR{$l=L+1$}
            \STATE ${\bf B_{L+1}} = \{k^M, k^A, k^N\}$;
            \STATE Sample each arm in $\{k^M, k^A, k^N\}$ for $\frac{T_1}{9}$ no. of times.\\
            \STATE Obtain $\hat{k}_{L+1}=\argmax\limits_{i \in \{k^M, k^A , k^N\}}\hat{\mu}^{L+1}_i.$
        \ENDFOR
       \STATE \textbf{Output:} $b_{\hat{k}_{T_1}} = \hat{k}_{L+1}$
    \end{algorithmic}
\end{algorithm}

\begin{rem}\label{rem:elimination}
Arms between $k^M\;\& \;k^A$ or $k^B\; \& \; k^N$ are eliminated in phase, and the arms between $k^A \;\& \;k^B$ always survive.
\end{rem}
After phase $l=1,2,\ldots, L$, $\lfloor\frac{2}{3}j_l\rfloor$ of the arms survive. For ease of exposition, we will drop the $\lfloor\rfloor$ function since this drop will influence only a few constants in the analysis. Thus, after the end of $L$ phases there will be three arms as 
\begin{align}
\label{eq:alg21}
     \left(2/3\right)^{L}K =3 \implies L =\frac{\log_2 K/3}{\log_2 3/2}.
\end{align}
Therefore, the {\it{UB3}} outputs the best beam as $\hat{k}_{L+1}$  which is equivalent to $b_{\hat{k}_{T_1}}$ after exploring for $T_1$ time slots. In the next section, we upper bound the error probability of {\it{UB3}}.

\section{Analysis}
\label{sec:analysis}
In this analysis, we find an upper bound and lower bound for the probability of best arm elimination for fixed-budget pure exploration bandit with unimodal structure. We first upper bound the error probability of  {\it{\ref{algo:UB3}}}. For analysis, we use the following assumption:
\begin{assu}\label{assump:dl}
 There exists a constant $D_L>0$ such that $|\mu_k-\mu_{k-1}|\ge D_L$ for $2\le k\le K$.
\end{assu} 

We note that this assumption is the same as that used in from \cite[Assumption 3.4]{ICML2011_UnimodalBandits} to analyze unimodal bandits in the regret minimization setting.

\subsection{Upper Bound for Algorithm UB3}
\begin{thm}
\label{thm:UB3_upper}
Let {\it UB3} is run for $T_1$ time slots in $L+1$ number of phases, where $L=\frac{\log_2 K/3}{\log_2 3/2}$ with output $\hat{k}_{L+1}$. Then, the probability that $\hat{k}_{L+1}$ output by {\it UB3} is not the best arm after $L+1$ phases is bounded as 
\begin{align}
\label{eqn:UpperboundUB3}
    & P(\hat{k}_{L+1}\ne k^*)\le  2\exp\left\{-\frac{T_1}{18}D_L^2\right\} + 2\exp{\left\{-\frac{T_1 K}{32}D_L^2\right\}}\nonumber\\
     &+2\exp\left\{-\frac{T_1 K}{72}D_L^2\right\} +2(L-2)\exp\left\{-\frac{T_1}{16}D_L^2\right\}.
\end{align}
\end{thm}
\begin{proof}
The proof is given in Appendix A. 
\end{proof}

Observe that the dominant first and last terms in the upper bound do not depend on $K$. The error probability is thus of order $\mathcal{O}\left(\log_2 K\exp\left\{-\frac{T_1D_L^2}{16}\right\}\right)$, where the error exponent term $\left(\exp\left\{-\frac{T_1D_L^2}{16}\right\}\right)$ does not depend on $K$. The {\it Sequential Halving (Seq. Halv.)} algorithm is proposed in \cite{seqhalving} for non-unimodal (unstructured) bandits and provided an upper bound for the probability of not choosing the optimal arm. The error probability is shown to be  $O\left(\log_2(K)\exp\left\{-\frac{T_1}{\log(K)H_2}\right\}\right)$. This bound matches with the lower bound derived in \cite{carpentier2016tight} for unstructured bandits up to multiplicative factor of $\log_2(K)$, where $H_2 = \max\limits_{k \neq k^*}\frac{k}{\Delta_k^2}$ is the complexity parameter depended upon the sub-optimality gap. Note that the error exponent in this bound of {\it Seq. Halv.} has $\log_2(K)$ factor. By exploiting the unimodal property, we shove off this factor. We next consider the lower bound for fixed-budget pure exploration with the unimodal structure which confirms that error exponent should be indeed independent of number of beams for any optimal algorithm.

\subsection{Lower bound for pure exploration unimodal bandit}
A lower bound on the probability of error for the fixed budget without assuming any structure is established in \cite{carpentier2016tight}. We adapt the proof to include the unimodal structure to derive a lower bound. To this end, we first define a set of bandit instances as follows.

Let us consider $K$ arms that follow the unimodal structure. Let $\{p_k\}_{1\leq k \leq K}$ be $K$ real numbers in the interval $[1/4,1/2]$ with $p_{k^*} = \frac{1}{2}$ and $p_1 \leq p_2 \leq \dots \leq p_{k^*-1} \leq p_{k^*} \geq p_{k^*+1} \geq \dots \geq p_K.$ For any $1 \leq k \leq K$ we consider the distribution $\nu_k$ to be Bernoulli distribution with mean $p_k$, i.e. $\nu_k := \text{Ber}(p_k)$. We consider another distribution $\nu'_k$ as Bernoulli distribution with mean $1-p_k$, i.e., $\nu'_k := \text{Ber}(1-p_k)$ with mean $1-p_k$. 

We define $K$ bandit problem as following. For $i \in \{1,\dots,K\}$ define the product distributions $\mathcal{G}^i := \nu_1^i \otimes \nu_2^i \otimes \dots \otimes \nu_K^i$ where 
\begin{align*}
    \nu_k^i :=    \begin{cases}
    &\nu_k 1\{k \neq i\} + \nu'_k 1\{k=i\},  \text{ if } k\in \{k^*-1,\\
    &\hspace{45mm} k^*,k^*+1\}\\
    &\nu_k,  \text{ otherwise}\end{cases}
\end{align*}
where $1\{A\}$ denotes the indicator function.
It is easy to note that only bandit instances $\mathcal{G}^i$, where $i \in \{k^*-1, k^*, k^*+1\}$ satisfy the unimodality structure and the not the others. Flipping of the reward for all other arms will result in a non-unimodal problem. Thus unlike \cite{carpentier2016tight} we have 3 bandit problems in the neighbourhood of $k^*$.

We define $d_k := p_{k^*} - p_k = \frac{1}{2} - p_k,$ for any $1 \leq k \leq K$. Set $\Delta^i_k = d_i + d_k, \text{ if } k \neq i \text{ and } \Delta^i_i = d_i,$ for any $i \in \{k^*-1, k^*+1\}$ and any $k \in  \{1,\dots,K\}.$ Note that $\{\Delta^i_k\}_k$ denotes the arm gaps of the bandit problem $i$. We also define for any $i \in \{k^*-1, k^*+1\}$ the quantity $$\bar{H}(i) := \sum\limits_{k \in \{i-1,i+1\}}\frac{1}{(\Delta^i_k)^2} \text{ and } \bar{h} = \sum\limits_{i \in \{k^*-1, k^*+1\}} \frac{1}{d^2_i \bar{H}'(i)}.$$ 
Theorem 2 from \cite{carpentier2016tight} can be rephrased in this setting as follows.
\begin{thm}
\label{thm:unimodal_lower}
For any bandit strategy that returns the arm $\hat{k}_{T_1}$ at time $T_1$, it holds that
\begin{align}\label{eq:thm1}
    \max\limits_{i \in \{k^*-1, k^*+1\}} &P_i(\hat{k}_{T_1} \neq i) \nonumber\\
    &\geq \frac{1}{6}\exp\left(-60\frac{T_1}{\bar{H}(k^*)} - 2\sqrt{T_1\log(18T_1)}\right),
\end{align}
and also 
\begin{align}\label{eq:thm2}
    \max\limits_{i \in \{k^*-1, k^*+1\}} &P_i(\hat{k}_{T_1} \neq i) \nonumber\\
    &\geq \frac{1}{6}\exp\left(-60\frac{T_1}{\bar{h}\bar{H}(i)} - 2\sqrt{T_1\log(18T_1)}\right).
\end{align}
\end{thm}
The proof of this theorem follows the lines similar to \cite[Thm. 2]{carpentier2016tight} after applying the change of measure rule to on the restricted set of arms. We skip the details.

\begin{cor}
\label{cor:unimodal_lower_cor}
Assume that 
$$T_1 \geq \max\bigg(\bar{H}(k^*),\bar{H}(i)\bar{h}\bigg)^2\frac{4 \log(6T_1K)}{(60)^2}.$$ For any bandit strategy
that returns the arm $\hat{k}_{T_1}$ at time $T_1$, it holds that 
$$\max\limits_{i \in \{k^*-1, k^*+1\}}P_i(\hat{k}_{T_1} \neq i) \geq \frac{1}{6}\exp\left(-120\frac{T_1}{\bar{H}(k^*)}\right),$$
and also $$\max\limits_{i \in \{k^*-1, k^*+1\}}P_i(\hat{k}_{T_1} \neq i) \geq \frac{1}{6}\exp\left(-120\frac{T_1}{\bar{H}(i)\bar{h}}\right).$$
\end{cor}
We can establish a lower bound using this corollary.

\begin{thm}
\label{thm:lower_bound}
For any unimodal bandit strategy  that returns arm $\hat{k}_{T_1}$ at time $T_1$, 
\begin{align}
\max\limits_{i \in \{k^*-1, k^*+1\}}P_i(\hat{k}_{T_1} \neq i) &\geq \frac{1}{6}\exp\left(-75\frac{T_1}{\bar{H}(i)}\right).
\end{align}
\end{thm}
\begin{proof}
The proof is given in Appendix B.
\end{proof}
We see that unlike in the case of the lower bound found in \cite{carpentier2016tight}, the lower bound of the error probability is not dependent on $\log_2(K)$. In addition, the complexity factor depends only on the sub-optimality gap between the optimal arm and its neighbours. This observation is similar to the lowed bound on cumulative regret established in \cite{ICML2014_UnimodalBandits}.

\section{Numerical Simulations}
\label{sec:experiments}
In this section, we corroborate our theoretical results using simulations. We first describe the simulation setup with parameters used and  present the results in the following subsections.

\subsection{Simulation Parameters}
We use the IEEE $802.11$ad system with parameters as described in \cite{hba} for numerical simulations. The carrier frequency ($f$) is set at $60$ GHz and the bandwidth is set at $2.16$ GHz. The transmit power $P=50$ dBm is shared between $K$ antennas which vary from $16$ to $128$. For the line of sight (LOS) path, we have the path loss model as
\begin{align}
\label{eqn:path_loss}
    PL(dB)=-27.5+20 \log_{10}(f)+10 \alpha \log_{10}(d)+\chi,
\end{align}
where $d$ is the transmission distance, the path loss exponent $\alpha$ is taken as $1.74$, and $\chi$ is the shadow fading component which follows Normal distribution with zero mean and $2$ dB variance.  In \eqref{eqn:path_loss} $f$ is in MHz and $d$ is in meters. Depending upon the beam selected the signal strength varies in $[-80,-20]$ dBm. The algorithm parameter are kept at $\rho_1=3$, $\gamma=0.5$ and $\zeta=0.1$. The channel parameters for the simulations are given in Table \ref{tab:parameters}. Simulation results are averaged over $1000$ iterations and confidence intervals are shown (when significant).
\begin{center}
\begin{table}
\centering
\begin{tabular}{c c } 
\hline
 \hline
{\bf Parameter} & {\bf Value}  \\ [0.5ex] 
 \hline
Carrier frequency ($f$) & $60$ Ghz\\  

Bandwidth (W) & $2.16$ GHz\\

Noise spectrum density ($N_0$) & $-174 \text{ dBm/Hz}$  \\
 Shadow fading variance (log-normal $\sigma$) & 2 dB \\ 
 Number of beams ($N$) & 16-128\\
 HBA parameters ($\rho_1,\gamma,\zeta$) & $(3,0.5,0.1)$ \\
 Distance considered (d) & $(20,40, 60,80)$ m\\
 Path loss exponent & $1.74$\\
 \hline
 \hline
\end{tabular}
 \caption{Parameters for simulation}
  \label{tab:parameters}
  \end{table}
\end{center}
We compare the performance of {\it \ref{algo:UB3}} with the following algorithms:
\begin{itemize}
    \item {\bf Sequential Halving (Seq. Halv.) \cite{seqhalving}: } This algorithm is used for pure exploration in non-unimodal bandits for a fixed $T_1$ time slots. The algorithm was proved to be optimal \cite{carpentier2016tight}, and hence a comparison would give the idea about, how the additional information of unimodality would improve the performance. 
    \item {\bf Linear Search Elimination (LSE) \cite{ICML2011_UnimodalBandits}: } Although this algorithm was proposed for continuous arm unimodal bandit problems, we have considered the algorithm for fixed $T_1$ time slots and for discrete arms. A comparison of {\it UB3} with {\it LSE} is pertinent as it is a well-known algorithm for unimodal bandits. 
    \item {\bf Hierarchical Beam Alignment (HBA) \cite{hba}: }This algorithm was shown to have good performance for regret minimization, when compared to existing algorithms, considering the prior knowledge of channel fluctuations. Our comparison with {\it HBA} will of throughput comparison for the period after the best beam has been identified.
    \item {\bf Hierarchical Optimal Sampling of Unimodal
Bandits (HOSUB) \cite{ICC2021_HOSUB}: } This algorithm exploits the benefits of hierarchical codebooks
and the unimodality of RSS to achieve the best arm in the fixed $T_1$ time slots. Our comparison with {\it HOSUB} will of throughput comparison for the period after the best beam has been identified.
\end{itemize}
We did not include {\it UBA} algorithm \cite{INFOCOM2018_EfficientBeamAllignment} for comparison since {\it HBA} was shown to have better performance for beam identification. Also, the {\it ATS} algorithm is not compared against as it is designed to minimize the cumulative regret and does not stop the exploration process. We note the computational complexity of complexity {\it HBA} scales quadratically in $T_1$ whereas it is of $O(T_1)$ in {\it UB3}, where $T_1$ is the duration of IA phase. In addition, {\it HBA} requires prior knowledge of channel fluctuations, i.e, the variance of the noise parameter, which is not required in {\it UB3}.
\begin{figure*}[t]
	\centering
	\begin{subfigure}{0.31\textwidth}
		\includegraphics[scale=.4]{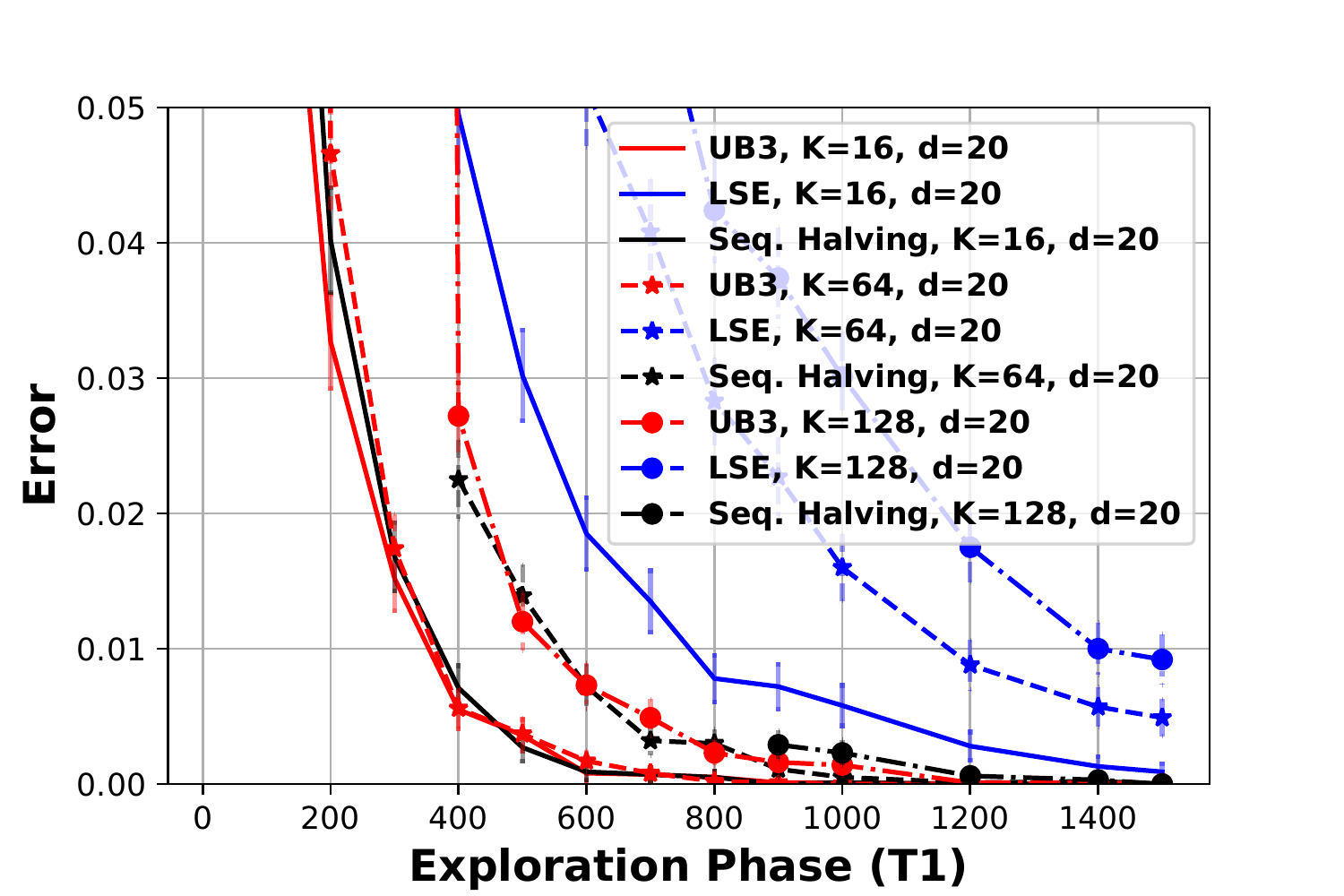}
    	\caption{Error Probability vs $T_1$, $d = 20$.}
    	\label{fig:err_vs_T1_K_d_20}
	\end{subfigure}
	\begin{subfigure}{0.31\textwidth}
		\includegraphics[scale=.4]{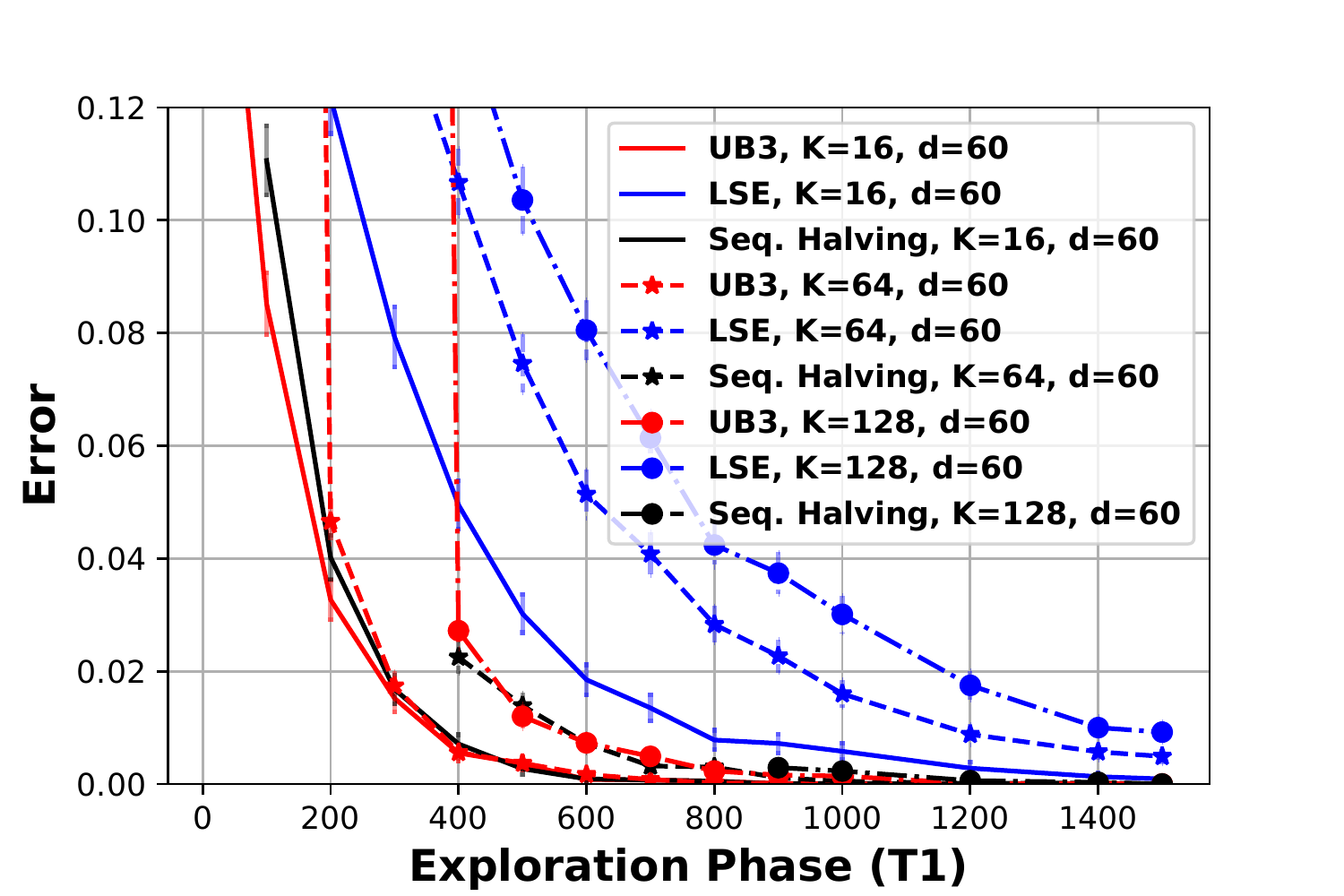}
    	\caption{Error Probability vs $T_1$, $d = 60$.}
    	\label{fig:err_vs_T1_K_d_60}
	\end{subfigure}
		\begin{subfigure}{0.31\textwidth}
		\includegraphics[scale=.4]{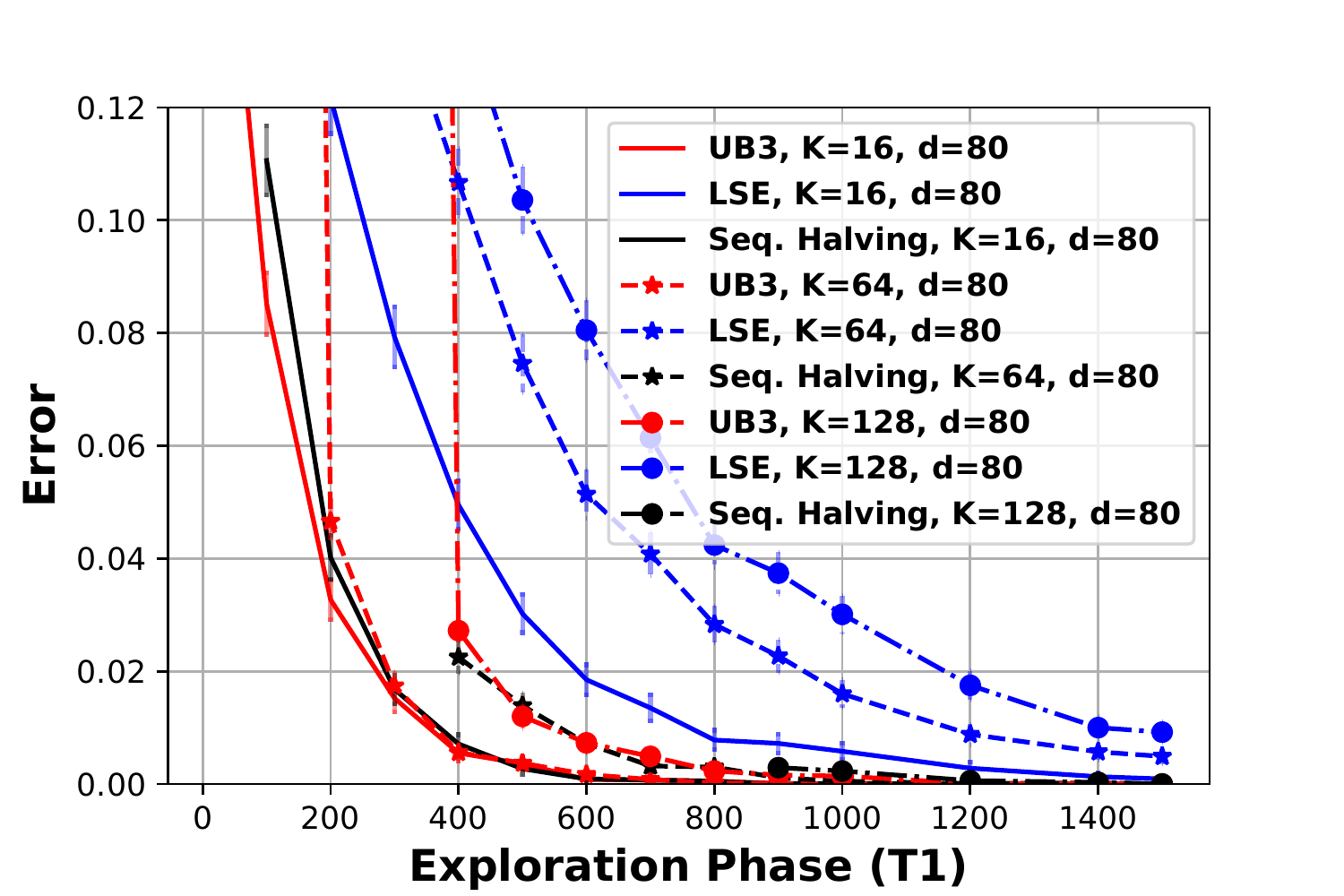}
    	\caption{Error Probability vs $T_1$, $d = 80$.}
    	\label{fig:err_vs_T1_K_d_80}
	\end{subfigure}
	\caption{Error performance of {\it UB3} vs $T_1$ for different no. of beams ($K$) and distance ($d$).}
	\label{fig:err_T1}
\end{figure*}
\begin{figure*}[t]
	\centering
	\begin{subfigure}{0.31\textwidth}
		\includegraphics[scale=.4]{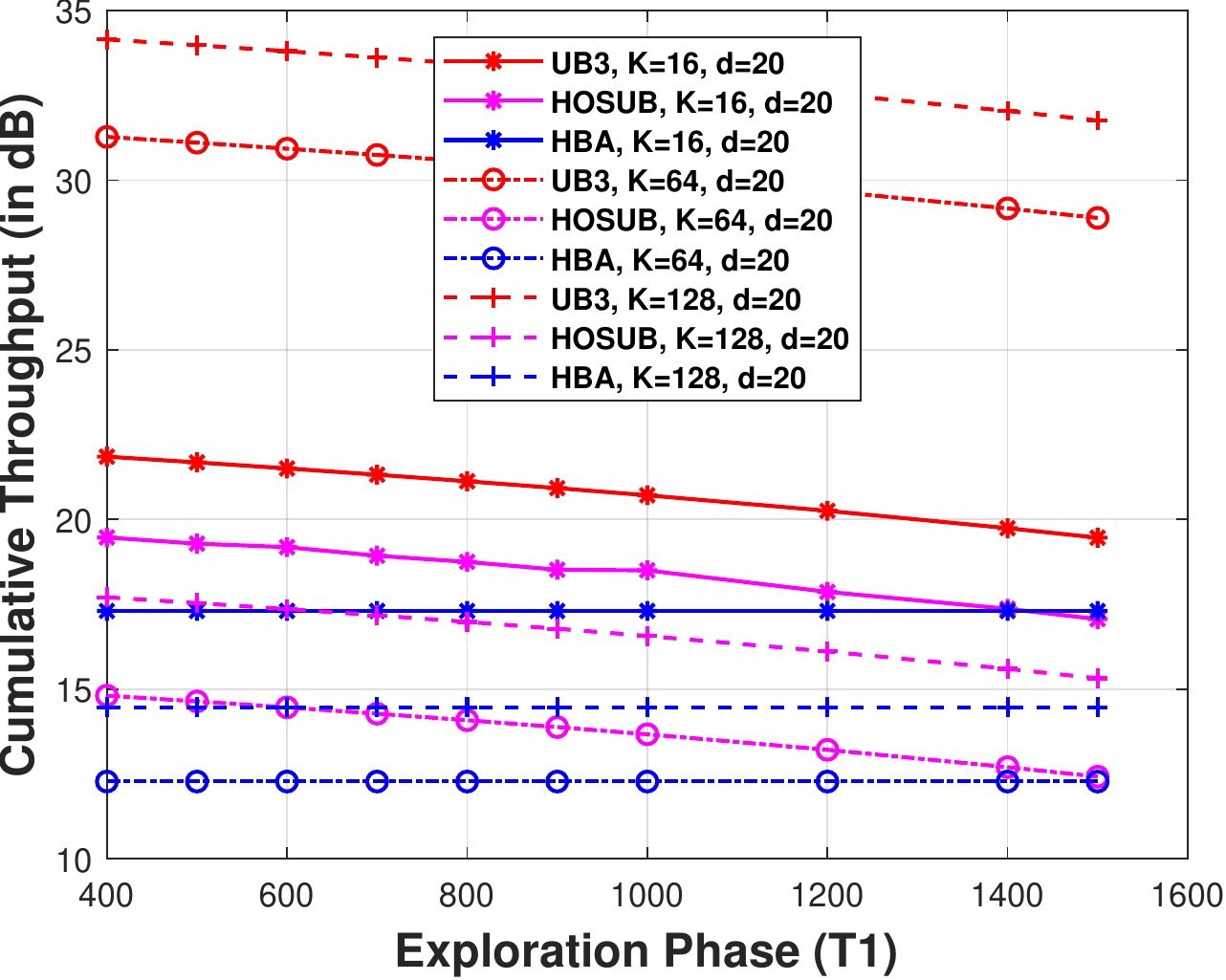}
    	\caption{Throughput vs $T_1$, $d = 20$.}
    	\label{fig:Throughput_vs_T1_K_d_20}
	\end{subfigure}
	\begin{subfigure}{0.31\textwidth}
		\includegraphics[scale=.4]{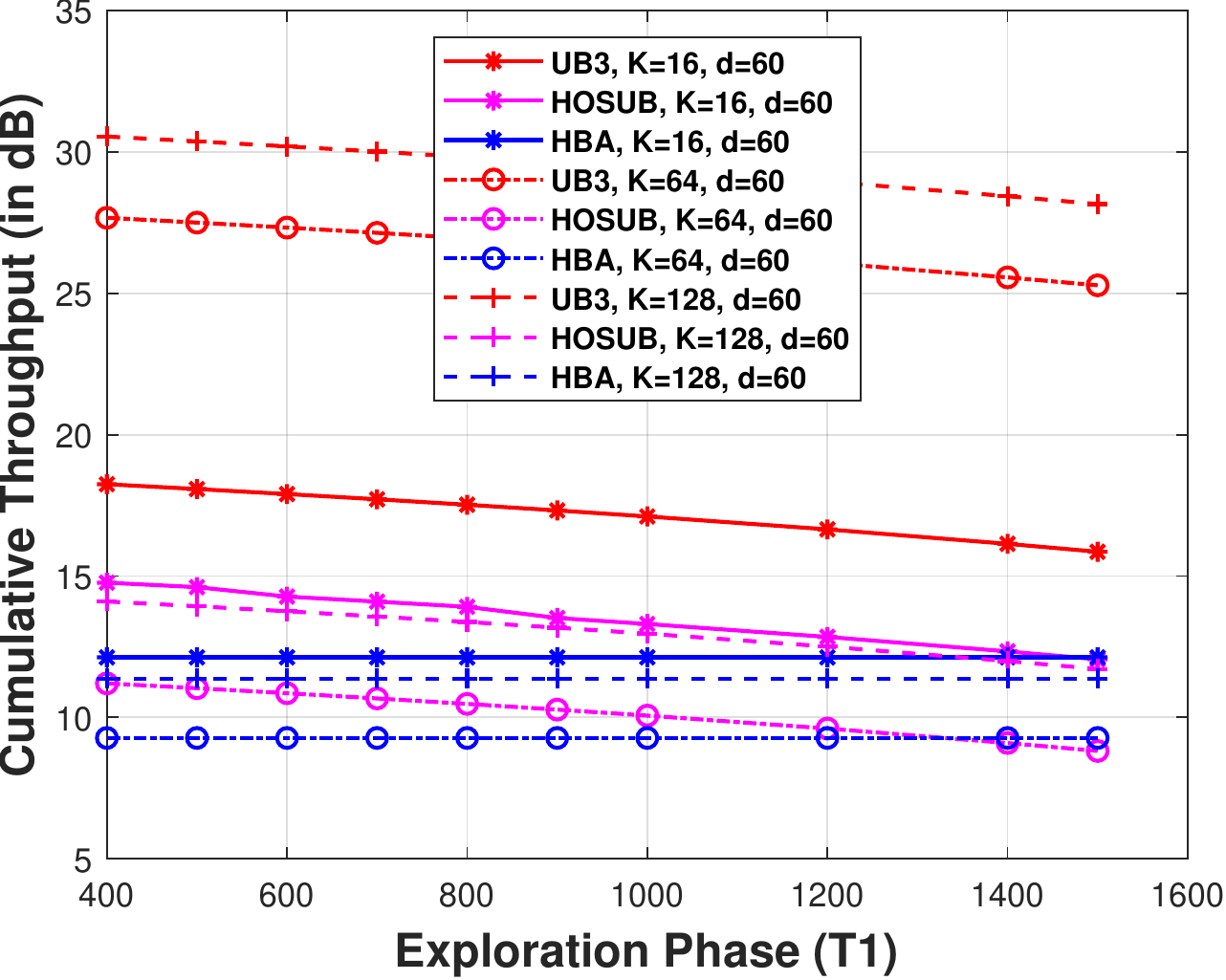}
    	\caption{Throughput vs $T_1$, $d = 60$.}
    	\label{fig:Throughput_vs_T1_K_d_60}
	\end{subfigure}
		\begin{subfigure}{0.31\textwidth}
		\includegraphics[scale=.4]{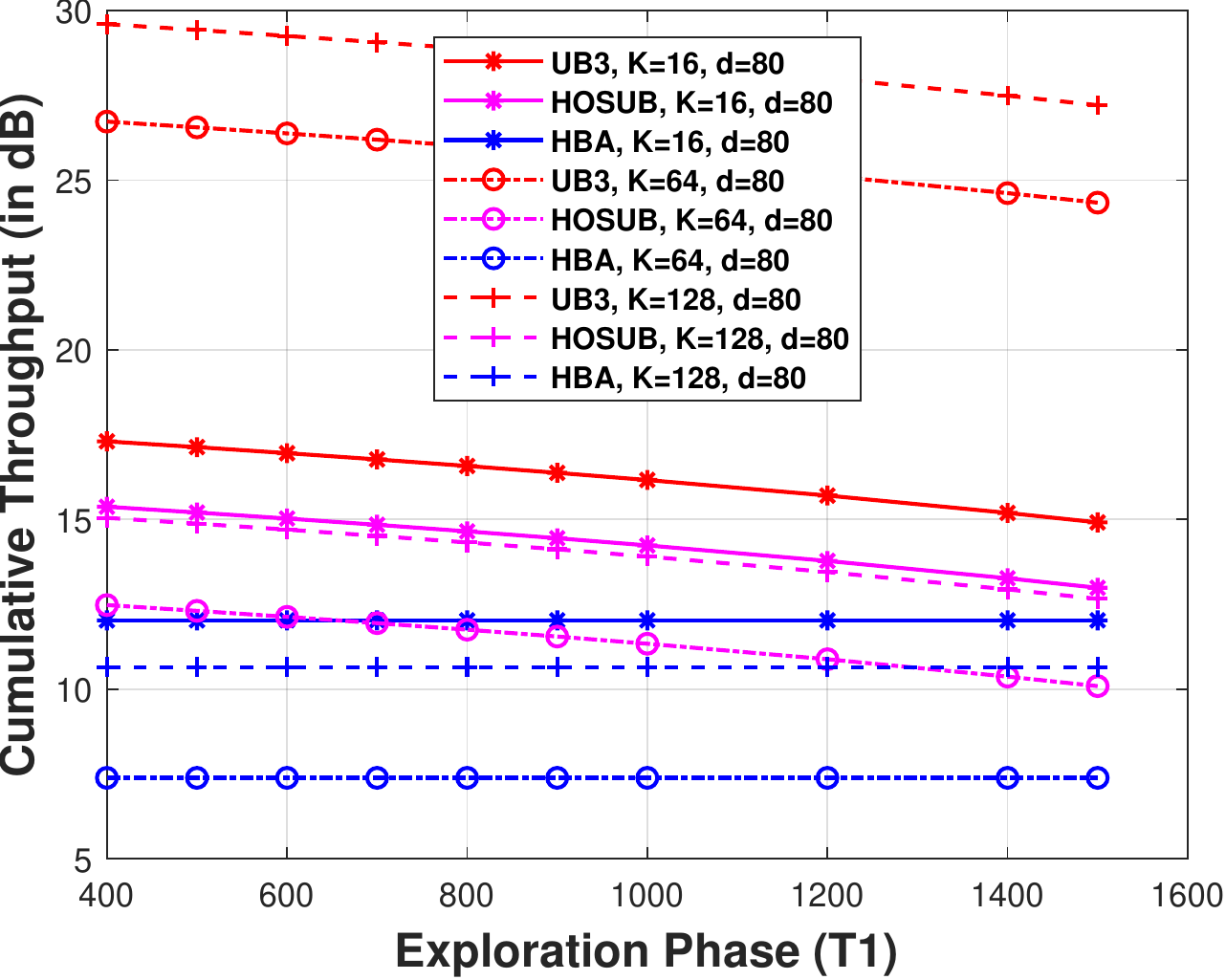}
    	\caption{Throughput vs $T_1$, $d = 80$.}
    	\label{fig:Throughput_vs_T1_K_d_80}
	\end{subfigure}
	\caption{Throughput performance of {\it UB3} vs $T_1$ for different no. of beams ($K$) and $d$.}
	\label{fig:throughput_T1}
\end{figure*}
\begin{figure*}[t]
\vspace{-3mm}
	\centering
	\begin{subfigure}{0.31\textwidth}
		\includegraphics[scale=.4]{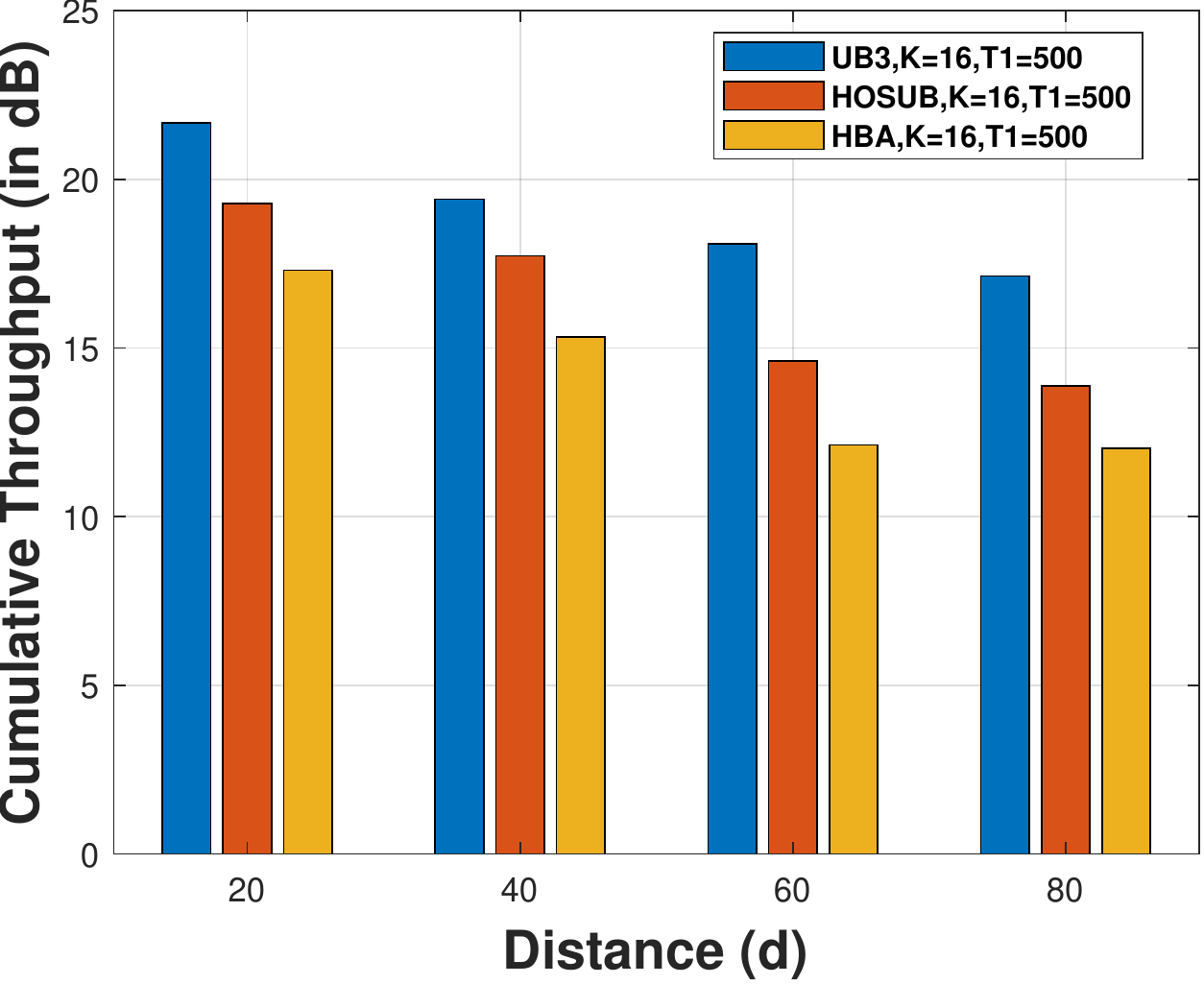}
    	\caption{Throughput vs $d, K=16, T_1 = 500$}
    	\label{fig:Throughput_vs_d_K_16}
	\end{subfigure}
	\begin{subfigure}{0.31\textwidth}
		\includegraphics[scale=.4]{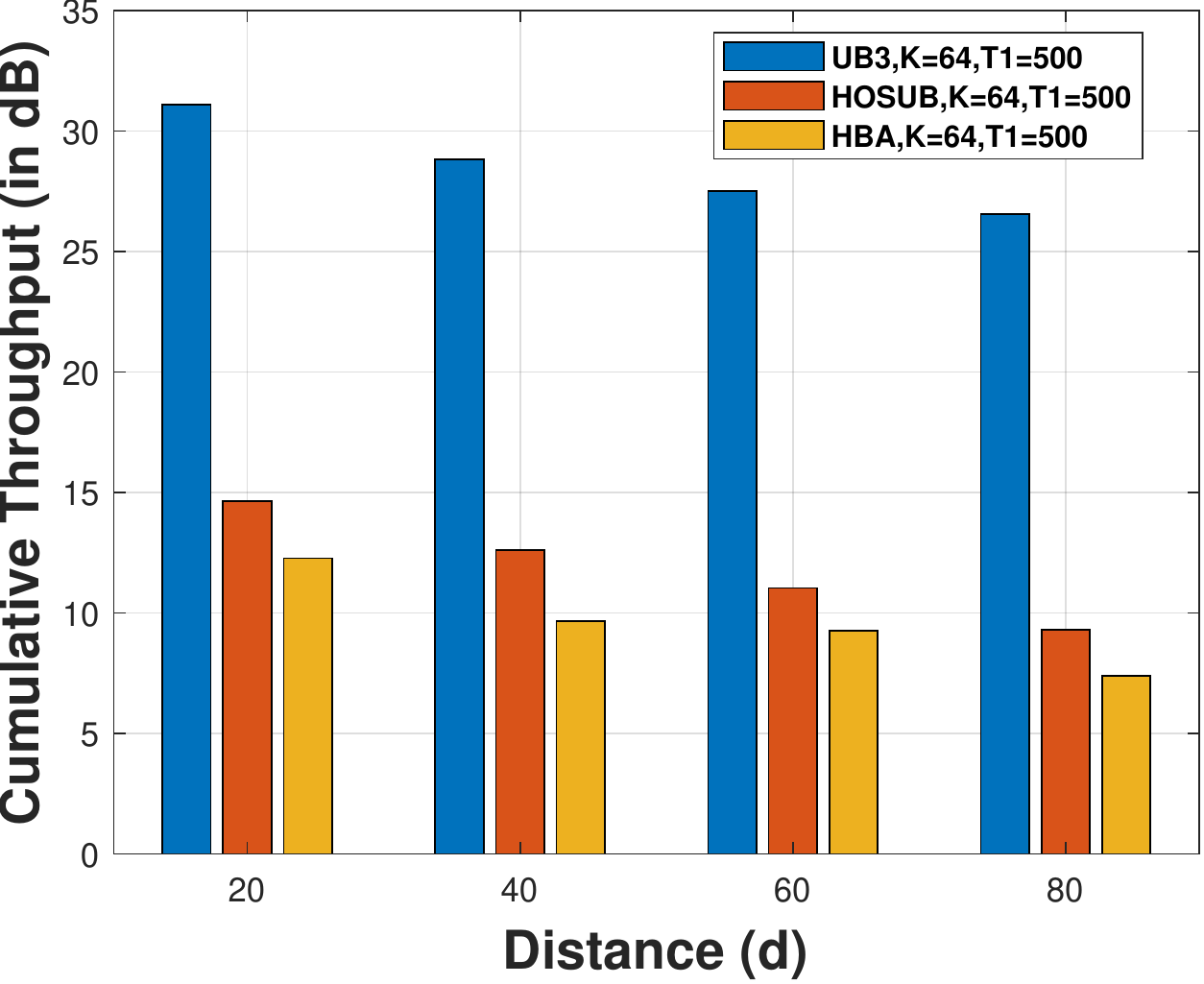}
    	\caption{Throughput vs $d, K=64, T_1 = 500$}
    	\label{fig:Throughput_vs_d_K_64}
	\end{subfigure}
		\begin{subfigure}{0.31\textwidth}
		\includegraphics[scale=.4]{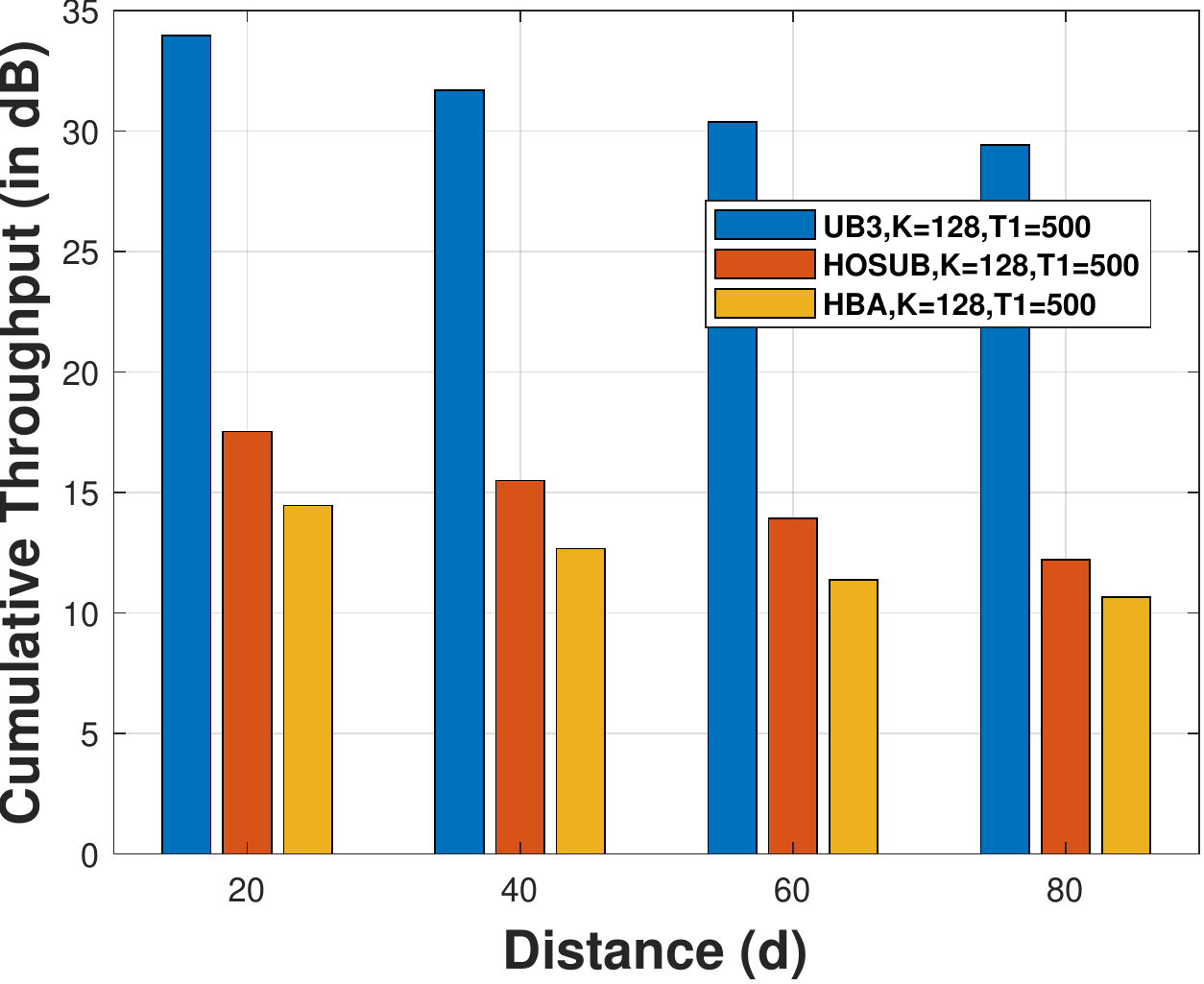}
    	\caption{Throughput vs $d, K=128, T_1 = 500$}
    	\label{fig:Throughput_vs_d_K_128}
	\end{subfigure}
	\caption{Throughput performance of {\it UB3} vs $d$ for different no. of beams ($K$) for $T_1=500$.}
	\label{fig:throughput_d}
\end{figure*}
\begin{figure*}[t]
	\centering
	\begin{subfigure}{0.31\textwidth}
		\includegraphics[scale=.4]{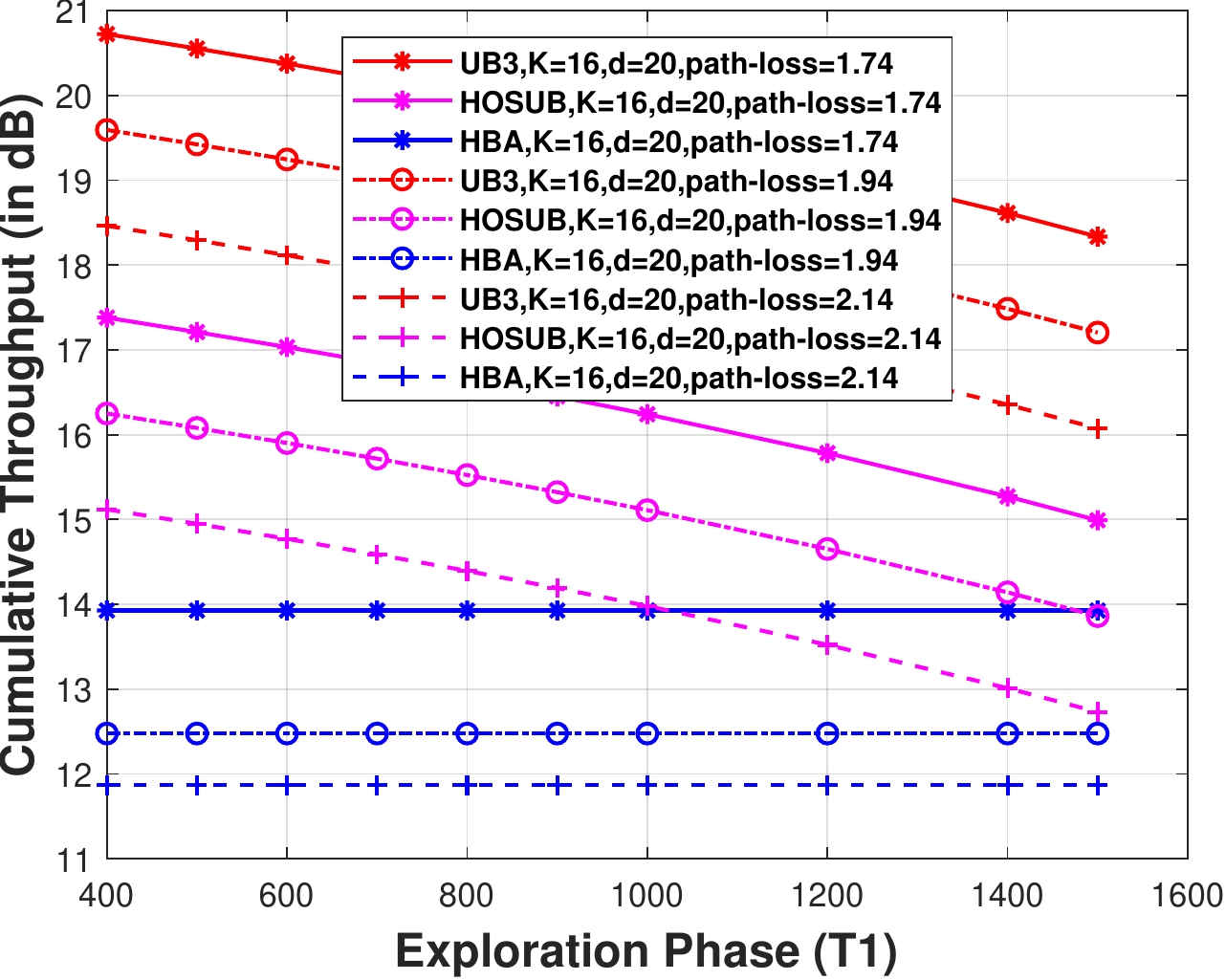}
    	\caption{Throughput vs $T_1, K=16, d=20$.}
    	\label{fig:Throughput_vs_T1_pathloss_K_16}
	\end{subfigure}
	\begin{subfigure}{0.31\textwidth}
		\includegraphics[scale=.4]{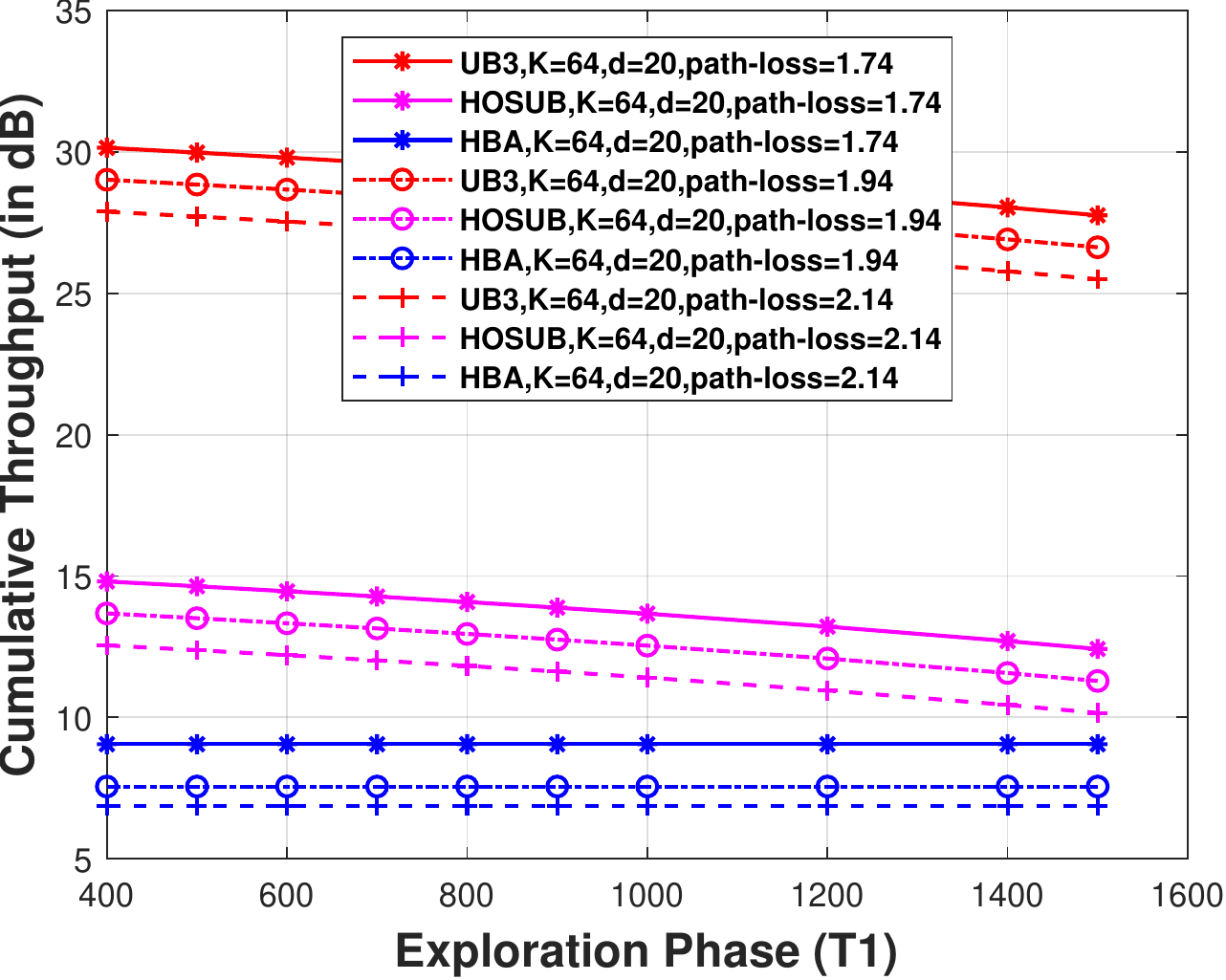}
    	\caption{Throughput vs $T_1, K=64, d=20$.}
    	\label{fig:Throughput_vs_T1_pathloss_K_64}
	\end{subfigure}
		\begin{subfigure}{0.31\textwidth}
		\includegraphics[scale=.4]{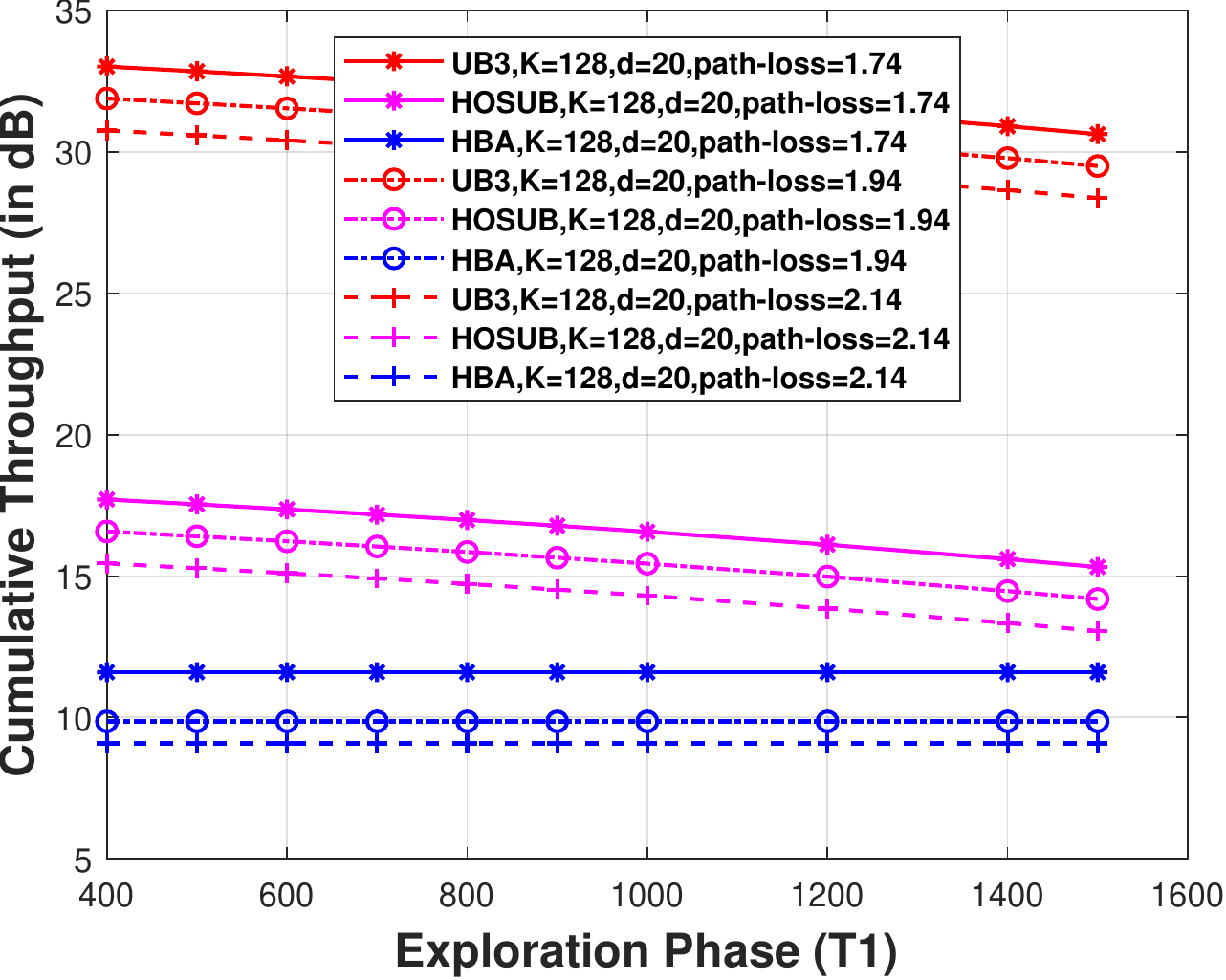}
    	\caption{Throughput vs $T_1, K=128, d=20$.}
    	\label{fig:Throughput_vs_T1_pathloss_K_128}
	\end{subfigure}
	\caption{Throughput performance of {\it UB3} vs $T_1$ for different path-loss exponent $\alpha$ and $d=20$.}
	\label{fig:throughput_pathloss}
\end{figure*}

\subsection{Comparison with other pure exploration algorithms}
We define the probability of error as the probability of not identifying the best beam after sampling for $T_1$ number of time slots. We consider $T_1$ as the exploration (IA) phase. We first compare the probability of error for {\it UB3} algorithm with {\it LSE} and {\it Seq. Halv.} which are used for pure exploration. We compare the probability of error performance of the algorithms for arm sizes of $K=\{16, 64,128\}$. The probability of error against the exploration time slots ($T_1$) is shown in Fig. \ref{fig:err_T1}. {\it LSE}, which has an equal number of sampling for the arms in all phases, has the worst probability of error performance than {\it Seq. Halv.}. The solid lines in the figure are for $K=16$, the dashed lines are for $K=64$ and the dot-dashed lines are for $K=128$. The comparisons are done for distance of $d=20,60$ and $80$ m.  The probability of error increases as we increase the beam size. However, for a small number of arms, both {\it Seq. Halv.} and {\it UB3} have comparable performance, but as the number of beams increases, {\it UB3} has a lesser probability of error compared to {\it Seq. Halv.} as evident from the case of $K=64$ and $K=128$. However, {\it{UB3}} can identify the best beam with a probability of more than 95\% within $100$ time slots for $16$ number of beams, while the other state-of-the-art algorithms take need at least $200$ time slots for executions. Note that for {\it Seq. Halv.} require at least $200, 400, 900$ rounds for $K=16,64,128$, respectively, to complete their execution hence their graph start after that many slots.  Moreover, as distance increases the probability of error of all the algorithms increases for each beam size. 

\par It is to be noted that even though the minimum time slots requirement (as a function of $K$) for {\it LSE} is much smaller than both {\it UB3} and {\it Seq. Halv.} for its feasible execution, the number of samples it runs for arms neighbouring to $k^*$ is much lesser resulting in more probability of error. {\it Seq. Halv.} needs at least $K\log_2(K)$ number of time slots to complete one phase and has samples for all arms in every phase. Hence, it has much fewer time slots remaining when the algorithm is executed in the neighborhood of $k^*$ as compared to {\it UB3}. Thus, the minimum time slots requirement for {\it UB3} as a function of $K$ is much lesser than that of {\it Seq. Halv.}, in addition to having a better probability of error performance. This demonstrates the advantage of exploiting the unimodality of the reward function.

\subsection{Comparison of throughput performance}
In this subsection we compare the throughput performance of {\it UB3} with the {\it HBA} and {\it HOSUB} algorithms. We look at the mean cumulative throughput, which is defined as the product of the mean value of the selected beam at the end of exploration normalized with the mean power of the best beam and the remaining available time slots, i.e,
\begin{align*}  \text{Throughput}=\mu^{\text{norm}}_{b_{L+1}}\times (T-T_1),
\end{align*}
where $T$ is the total available time slots and $T_1$ is the number of time slots available for exploration of the best beam. Note that {\it HBA} will not have a fixed $T_1$, and hence we will find the throughput after the expected exploration time $E(T_1)$, obtained as the average of many runs. Thus the throughput will be fixed for {\it HBA} for a fixed $T$, while it will vary for varying $T_1$ for {\it UB3} and {\it HOSUB}. Fig. \ref{fig:throughput_T1} compares the cumulative throughput of {\it UB3}, {\it HBA} and {\it HOSUB} for different $K$ and distances of $d=20,60$ and $80$ m for $T = 3000$ time slots.

As the number of beams (arms) increases, the beam becomes narrower, and hence the reward for the best beam also increases. In that case, {\it HOSUB} requires more exploration time slots to learn the optimal beam, thereby exploiting sub-optimal beams in the data transmission phase for the given $T_1$ time slots. On the other hand, for {\it HBA}, since the average time required to find the best arm increases, the average throughput will decrease as the number of arms increases. This increasing and decreasing of throughput for {\it HBA} is seen in Fig. \ref{fig:throughput_T1} and Fig. \ref{fig:throughput_d}. However, the {\it UB3} algorithm is not much affected by the number of increases in arms for finding the best arm, and hence the throughput will only increase with increasing gain for the best arm. {\it UB3} can improve the throughput by more than 45\% compared to {\it HBA} and by more than 15\% compared to {\it HOSUB}. This too is evident from both Fig. \ref{fig:throughput_T1} and Fig. \ref{fig:throughput_d}. However, the throughput will decrease as we increase the transmission distance, refer Fig. \ref{fig:throughput_d}.

Finally, we compare the throughput performance for varying path loss exponent given as $\alpha\in\{1.74,1.94,2.14\}$, as shown in Fig. \ref{fig:throughput_pathloss}. The path loss exponent increases when there are more barriers; for example when the receiver moves from outdoor to indoor. The throughput indeed decreases for all {\it HBA}, {\it HOSUB} and {\it UB3}, but {\it UB3} still outperforms {\it HBA} and {\it HOSUB}.

\section{Conclusion and future work}
\label{sec:discussion}

We investigated the problem of beam alignment in mmWave systems using the multi-armed bandits (MAB). While earlier works used the cumulative regret minimization setting to learn the best arm, we used the fixed-budget pure-exploration setting framework exploiting the unimodal structure of the received signal strength of the beams. We developed an algorithm named \ref{algo:UB3} that identified the best beam with high probability. We gave an upper bound on the error probability of UB3 and established that it is optimal. Simulations validated the efficiency of {\it UB3} which can identify the best beam using a smaller number of explorations that can translate to improvement in throughput by more than 15\% compared to other state-of-art algorithms. Due to its simple structure, {\it UB3} is easy to implement and comes with a lower computational complexity -- {\it UB3} has a computational complexity of $\mathcal{O}(T)$, whereas it is $\mathcal{O}(T^2)$ for {\it HBA} \cite{hba}. The {\it UBA} algorithm in \cite{INFOCOM2018_EfficientBeamAllignment} needs to solve a convex optimization problem in each time which is expensive.

{\it UB3} works well when only the LOS path is present and the RSS of beams satisfies the unimodal property. However, when NLOS paths are present, we are faced with multimodal functions. {\it UB3} can be adapted to handle the multi-modal functions by using backtracking ideas proposed in \cite{backtrack}. In backtracking, eliminated arms are revisited to check if it is done by mistake and thus will not be stuck in a sub-optimal set of beams. It is interesting to evaluate the {\it UB3} algorithms with backtracking on multimodal function and establish its performance guarantees. 

\section{Appendix}
\label{sec:appendix}
In this section, we will provide proof of the main results. 

\subsection{Proof of Theorem \ref{thm:UB3_upper}}
\begin{proof}
{\it UB3} runs for $T_1$ horizon in $L+1$ number of phases that satisfies \eqref{eqn:L_T}, where $L=\frac{\log_2 K/3}{\log_2 3/2}$ and outputs the arm $\hat{k}_{L+1}$. We will now upper bound the probability of error as,
\begin{align}\label{eq:ubound}
    P(\hat{k}_{L+1}\ne b_{k^*})&=\sum_{l=1}^{L+1}P(b_{k^*} \text{ elim. in } l|b_{k^*} \text{ not elim. in }<l)\nonumber\\
    &\le \sum_{l=1}^{L+1}P(b_{k^*} \text{ elim. in } l).
\end{align}
The best arm is eliminated in phase $l$ in the following cases:
\begin{enumerate}
    \item $b_{k^*} \in \{k^M,\dots, k^A\}$, and $\hat{\mu}^l_{k^B}$ or $\hat{\mu}^l_{k^N}$ is greater than both $\hat{\mu}^l_{x^M}$ and $\hat{\mu}^l_{k^A}$
    \item $b_{k^*} \in \{k^B,\dots, k^N\}$, and $\hat{\mu}^l_{k^M}$ or $\hat{\mu}^l_{k^A}$ is greater than both $\hat{\mu}^l_{k^B}$ and $\hat{\mu}^l_{k^N}$
\end{enumerate}
The two cases are illustrated in Fig. \ref{case12}. From Remark \ref{rem:elimination}, $b_{k^*}$ will not get eliminated if $b_{k^*}\in \{k^A,\dots, k^B\}$. However  we will upper bound the probability of error by assuming that $b_{k^*}$ will always fall in the above two cases. Notice that Case 1 and Case 2 are symmetrical. Hence we can consider that $b_{k^*}$ will always fall in either one of the cases. Without loss of generality, we  consider Case 1. 
\begin{figure}[t]
    \centering
    \begin{tikzpicture}[line width=1.5pt]
        \draw (-0.2,0) --++(7,0) ; 
        \foreach \i in{0,2.2,4.4,6.6}{\draw [thin](\i,0.1)--++(0,-0.2);}
        \draw[red,dashed] (0, 0.2) .. controls (5.3,1.8) .. (6.6, 0.4);
        \draw[blue] (0, 0.4) .. controls (1.3,1.8) .. (6.6, 0.2);
        \draw[blue, dotted, line width=1] (1.5,1.3)--++(0,-1.3)node[below]{$b_{k^*}$};
        \draw[red, dotted, line width=1] (5.2,1.3)--++(0,-1.3)node[below]{$b_{k^*}$}; 
        \node at (0,-0.3) {{$k^M$}};
        \node at (2.2,-0.3) {{$k^A$}};
        \draw [decorate,decoration={brace,mirror,amplitude=10pt},yshift=0pt,line width=1pt] (2.2,-0.5) -- (4.4,-0.5) ;
        \node at (3.3,-1.1) {\footnotesize$\frac{j_l}{3}D_L$};
        \node at (4.4,-0.3) {{$k^B$}};
        \node at (6.6,-0.3) {{$k^N$}};
        \node[blue] at (0.2,1.7){Case 1};
        \node[red] at (6.2,1.7){Case 2};
    \end{tikzpicture}
    \caption{Different cases of elimination in any phase $l$. $b_{k^*}$ will not get eliminated if it is in between arms $k^A$ and $k^B$.}
    \label{case12}
     \vspace{-8mm}
\end{figure}
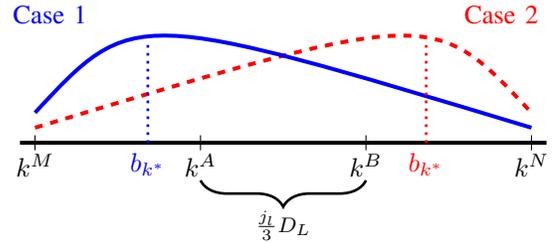
\begin{align}
    P(b_{k^*} &\text{ elim. in } l)\le P(b_{k^*} \text{ elim. in } l|b_{k^*}\in\{k^M,\dots,k^A\})\nonumber
\end{align}
\begin{align}
\label{eq:mu_b}
  &P(b_{k^*} \text{ elim. in } l) \nonumber\\
  &\le P(\hat{\mu}^l_{k^B}> \hat{\mu}^l_{k^M} \text{and}\,\, \hat{\mu}^l_{k^A}|b_{k^*}\in\{k^M,\dots,k^A\})\nonumber\\
    &\quad+P(\hat{\mu}^l_{k^N}> \hat{\mu}^l_{k^M} \text{and}\,\, \hat{\mu}^l_{k^A}|b_{k^*}\in\{k^M,\dots,k^A\})\nonumber\\
    &\le 2P(\hat{\mu}^l_{k^B}> \hat{\mu}^l_{k^M} \text{and}\,\, \hat{\mu}^l_{k^A}|b_{k^*}\in\{k^M,\dots,k^A\}),
\end{align}
where the last inequality is due to the fact that, for Case 1, $\mu_{k^B}\ge\mu_{k^N}$ by unimodality. Now for Case 1, $\mu_{k^A}$ is always greater than $\mu_{k^B}$, but $\mu_{k^M}$ may not be greater than  $\mu_{k^B}$. Then, we can further upper bound \eqref{eq:mu_b} as
\begin{align}
\label{eqn:prob_ub}
     P(b_{k^*} \text{ elim. in } l)\le 2P(\hat{\mu}^l_{k^B}> \hat{\mu}^l_{k^A}|b_{k^*}\in\{k^M,..,k^A\}).
\end{align}
Applying Hoeffding's inequality in \eqref{eqn:prob_ub}, we have
\begin{align}\label{eq:hoeff}
    P(\hat{\mu}^l_{k^B}> \hat{\mu}^l_{k^A})\le \exp\left\{-\frac{1}{2}\frac{N_l}{4}\left(\Delta_{A,B}\right)^2\right\},
\end{align}
where $\Delta_{A,B}=\mu_{k^A}-\mu_{k^B}$ which is greater than $0$ for Case 1. From Assumption \ref{assump:dl}, and the fact that there are at least $\frac{j_l}{3}$ arms between $k^A$ and $k^B$, for Case 1 we have, $\Delta_{A,B} \geq (j_l/3) D_L$.
Thus from \eqref{eq:mu_b} and \eqref{eq:hoeff} we have, 
\begin{align}
\label{eq:tlkl}
    P(b_{k^*} \text{ elim. in } l)\le 2\exp\left\{-\frac{N_l}{72}\bigg(j_lD_L\bigg)^2\right\}.
\end{align}
Using $j_l=\left(\frac{2}{3}\right)^lK$ in \eqref{eq:tlkl} we can find the probability of best arm getting eliminated in phase 1 and 2, phase $L+1$, and the rest of the phases separately. Using \eqref{eq:alg21}, we have
 \begin{align}
 \label{eq:round12}
    P(b_{k^*} \text{ elim. in } 1\& 2)\le & 2\exp{\left\{-\frac{T_1K}{32}D_L^2\right\}}\nonumber\\&
    +2\exp\left\{-\frac{T_1K}{72}D_L^2\right\}.
\end{align}
For phase $L+1$, since the best arm is selected among 3 arms when each arm is sampled $T_1/9$ times, we have
\begin{align}\label{eq:L1}
    P(b_{k^*}\text{ elim. in phase }L+1)\le 2\exp\left\{-\frac{T_1}{18}D_L^2\right\}.
\end{align}
From \eqref{eq:tlkl}, the error probability for the remaining phases is
\begin{align}\label{eq:Lrest}
     P(\text{best arm} & \text{ elim. in phase 3 to phase L})\nonumber\\&\le 2\sum_{l=3}^{L}\exp\left\{-\frac{T_1}{8}\frac{K^2}{9}\left(\frac{2}{3}\right)^{2(l-1)}\frac{2^{L-l+1}}{3^{L-l+2}}D_L^2\right\}\nonumber\\
     &\le  2\sum_{l=3}^{L}\exp\left\{-\frac{T_1K}{48}\left(\frac{2}{3}\right)^lD_L^2\right\}\nonumber\\
     &\le 2(L-2)\exp\left\{-\frac{T_1}{16}D_L^2\right\}.
\end{align}
By \eqref{eq:ubound}, \eqref{eq:round12}, \eqref{eq:L1} and \eqref{eq:Lrest}, we obtain the upper bound as
\begin{align}
     P(&\hat{k}_{L+1}\ne b_{k^*})\nonumber\\
     &\le 2\exp\left\{-\frac{T_1}{18}D_L^2\right\} +2\exp{\left\{-\frac{T_1K}{32}D_L^2\right)}\nonumber\\
     &+2\exp\left\{-\frac{T_1K}{72}D_L^2 \right\}+2(L-2)\exp\left\{-\frac{T_1}{16}D_L^2\right\}.\nonumber\qedhere
\end{align}
\end{proof}

\subsection{Proof of Theorem \ref{thm:lower_bound}}
\begin{proof}
We have, $p_k = \frac{1}{2}-d_k$ such that $p_k \in [1/4,1/2]$ and follows unimodality and $p_{k^*} = \frac{1}{2}$.  Upper bounding $\bar{h}$, we have
\begin{align}
    \bar{h} &= \sum\limits_{i\in\{k^*-1, k^*+1\}}\frac{1}{d^2_i\bar{H}(i)}\nonumber\\
    &= \frac{1}{d^2_{k^*-1}\bar{H}(k^*-1)} + \frac{1}{d^2_{k^*+1}\bar{H}(k^*+1)}\nonumber\\
    &= (I) + (II).
    \end{align}
    
We will upper bound (I) and (II),
\begin{align}
    &d^2_{k^*-1}\bar{H}(k^*-1) =  d^2_{k^*-1}\sum\limits_{k \in \{k^*-2,k^*\}}\frac{1}{(d_{k^*-1}+d_k)^2}\nonumber\\
    \intertext{Since $d_{k^*} = 0$ and $d_{k^*-2} \geq d_{k^*-1},$ we get}
    &d^2_{k^*-1}\bar{H}(k^*-1)\leq 1 + \frac{1}{4} = \frac{5}{4} \label{eqn:new_bound_1}\\
    &d^2_{k^*+1}\bar{H}(k^*+1) =  d^2_{k^*+1}\sum\limits_{k \in \{k^*, k^*+2\}}\frac{1}{(d_{k^*+1}+d_k)^2}\nonumber
\end{align}
Since $d_{k^*} = 0$ and $d_{k^*+2} \geq d_{k^*+1},$ we get
\begin{align}
    d^2_{k^*+1}\bar{H}(k^*+1)&\leq 1 + \frac{1}{4} = \frac{5}{4}. \label{eqn:new_bound_2}\\
    \intertext{By \eqref{eqn:new_bound_1} and \eqref{eqn:new_bound_2} we get} 
    \bar{h} \geq \frac{4}{5} + \frac{4}{5} = \frac{8}{5}.\nonumber\\
    \intertext{Putting the value of $\bar{h}$ in Corollary we get}
  \implies  \max\limits_{i \in \{k^*-1, k^*+1\}}P_i(\hat{k}_T \neq i) &\geq exp\left(-75\frac{T}{\bar{H}(i)}\right).\nonumber\qedhere
\end{align}
\end{proof}


\bibliographystyle{IEEEtran}
\bibliography{ref}

\end{document}